\documentclass{fundam}

\usepackage{amsmath}
\usepackage{mathptmx}
\usepackage{amssymb}
\usepackage{mathtools}
\usepackage{enumitem}

\usepackage{graphicx}

\usepackage{float}
\usepackage{tikz-cd}
\newcommand{\qed}{}
\usepackage{comment}

 \usepackage{hyperref}

\begin{document}

\setcounter{page}{363}
\publyear{2021}
\papernumber{2078}
\volume{182}
\issue{4}

  \finalVersionForARXIV
  %%\finalVersionForIOS

\title{Perpetual Free-choice Petri Nets are Lucent \\ Proof of a Theorem of van der Aalst Using $CP$-exhaustions}

\author{Joachim Wehler\thanks{Address  for correspondence:$\,$Department of
                          Mathematics,$\,$Ludwig-Maximilians-Universit{\"a}t$\,$M{\"u}nchen,$\,$\mbox{Theresienstrasse$\,$39}, D-80333 M\"unchen, Germany.\newline \newline
          \vspace*{-6mm}{\scriptsize{Received  September 2020; \ revised September 2021.}}}
 \\
Department of Mathematics \\
Ludwig-Maximilians-Universit{\"a}t M{\"u}nchen (LMU Munich)\\
M\"unchen, Germany\\
Joachim.Wehler@gmx.net}

\maketitle

\runninghead{J. Wehler}{Perpetual Free-choice Petri Nets are Lucent}

\setcounter{footnote}{0}

\begin{abstract}
Van der Aalst's theorem is an important result for the \mbox{analysis} and synthesis of process models. The paper proves the theorem by exhausting perpetual free-choice Petri nets \mbox{by $CP$-subnets}. The \mbox{resulting $T$-systems} are investigated by elementary methods.
\end{abstract}


%======================================================================================

\section{Introduction} \label{sect_introduction}

In the course of his work on the analysis and synthesis of process models van der Aalst introduced the concept of a \emph{lucent} process model \cite{Aal2018},  \cite{Aal2019}. The global state of a lucent process model is known when all actions that are possible in a given state are known. When expressed in the language of Petri nets lucency means: If two reachable markings of the Petri net enable the same transitions (actions) then the markings are equal, i.e. each state of the system is already determined by the set of its enabled transitions.

How to decide by inspection of a Petri net whether it is lucent? Van der Aalst considers the class of live and bounded Petri nets. He provides interesting examples from this class which are not lucent, even though they are safe. He names \emph{perpetual} Petri nets the class of live and bounded Petri nets with a "regeneration point", i.e. with a home marking which marks only the places of a distinguished \mbox{cluster. \cite[Theor. 3]{Aal2018}} states:
\begin{flushleft}\textbf{Theorem} [Van der Aalst's theorem on lucency]\label{theor_vanaalsttheorem}
Each perpetual free-choice system is lucent.
\end{flushleft}\bigskip
For the importance of lucency in the context of process discovery see \cite{Aal2019}. The proof of the theorem in \cite{Aal2018} has a gap as van der Aalst remarks \mbox{in \cite{Aal}}. To close the gap van der Aalst uploaded a revised version of a previous paper, see \cite[Theor. 3]{Aal2020}.
\footnote{Added in proof: See also ''van der Aalst, Wil M.P.: Free-Choice Nets With Home Clusters Are Lucent. \emph{Fundamenta Informaticae}, 2021 (in print).''}\bigskip

The purpose of the present paper is to give a proof of van der Aalst's theorem which uses some fundamental results from the theory of free-choice systems. Notably we focus on the existence \mbox{of $CP$-subnets} of well-formed free-choice nets. Hence our proof uses different ideas than those in \cite{Aal2018}, \cite{Aal2020}: We exhaust a well-formed free-choice net by a family of $CP$-subnets. In the end the problem reduces to a statement about lucency of certain perpetual $T$-systems. Here the claim can be proved by elementary methods. For a scheme of the proof in the present paper see Figure \ref{fig_planofproof} in Section \ref{sect_enablingequaivalencemarkingequality}.

%======================================================================================
\section{Basic concepts and results}\label{sect_basicconcepts}

To fix the notation and for the convenience of the reader we recall some basic concepts and results. As common in mathematics we denote set inclusion by ``$\subset$''. The sign covers both cases, proper inclusion and equality; we de not use the sign ``$\subseteq$''. The symbol "$\subset$" also the denotes the inclusion of subnets. Furthermore, we mostly follow the standard textbook about free-choice systems \cite{DE1995}.
\begin{remark}[Concepts, notations, basic results]\label{rem_concepts} A \emph{net}
$$N=(P,T,F)$$
is a bipartite, directed graph with nodes the set $P$ of \emph{places} and the set $T$ of \emph{transitions}, and the set of \emph{edges}
$$F \subset \left((P \times{} T) \cup (T \times{} P) \right).$$
All nets are finite. We use also the notation $N_P:=P,\ N_T:=T$. We represent edges by arrows, pointing from the first to the second component of the pair.

\begin{enumerate}
\item \emph{Structure}: Two nets $N_j=(P_j,T_j,F_j),\ j=1,2,$ are \emph{disjoint} if
$$P_1\cap P_2=T_1 \cap T_2 =\emptyset.$$
The \emph{disjoint union} of a family of pairwise disjoint nets is the union of these nets. A~net  \mbox{$N^\prime=(P^\prime, T^\prime, F^\prime)$} is a \emph{subnet} $N^\prime \subset N$ if
$$P^\prime \subset P, T^\prime \subset T, F^\prime \subset F$$
The subnet $N^\prime \subset N$ is a \emph{full subnet} if
$$F^\prime=F \cap \left((P^\prime \times{} T^\prime) \cup (T^\prime \times{} P^\prime) \right)$$
If not explicitly stated otherwise the term \emph{subnet} in the present paper means a full subnet. But we will also consider subnets which are not full subnets. Each pair of
\mbox{subsets $P^\prime \subset P,\ T^\prime \subset T$} generates a full subnet
$$span_N<P^\prime,T^\prime>\ \subset N$$
with node set $P^\prime \cup T^\prime$. For a full subnet $N^\prime=(P^\prime, T^\prime, F^\prime) \subset N$ the \emph{complement}
$$\overline{N^\prime} :=N\setminus N^\prime \subset N$$
is the full subnet \mbox{of $N$} spanned by the nodes from $(P \cup T) \setminus (P^\prime \cup T^\prime)$. A \mbox{subnet $N^\prime \subset N$} is \emph{transition-bordered} if its places $p \in (N^\prime)_P$ satisfy $^\bullet p \cup p ^\bullet \subset N^\prime$.\bigskip

A \emph{path} of $N$ is a non-empty sequence of nodes of $N$
$$\delta=(x_1,...,x_n) \text{ with } (x_i,x_{i+1}) \in F,\ i=1,...,n-1.$$
We use the notation $\delta \subset N$. The path $\delta$ is \emph{elementary} if $x_i\neq x_j$ for $1\leq i \neq j \leq n$. It is a \emph{circuit} if $(x_n,x_1) \in F$. The \emph{concatenation} of two adjacent paths $\delta_1=(x_	1....,x_n)$ and $\delta_2=(x_n,...,x_{n+k})$ is
$$\delta_1* \delta_2:=(x_1,...,x_n,x_{n+1},...,x_{n+k})$$
If $x_1=x_{n+k}$ then $\delta_1* \delta_2$ induces the circuit $(x_1,...,x_{n+k-1})$.\bigskip

The net $N$ is \emph{weakly connected} or just \emph{connected} when each two nodes $x, y \in N$ satisfy
$$(x,y) \in (F\cup F^{-1})^*\ \text{(symmetric, reflexive and transitive closure)}$$
The net is \emph{strongly connected} when each two nodes $x, y \in N$ can be joined by a \mbox{path $(x,...,y)$} leading \mbox{from $x$ to $y$}, i.e. $(x,y)\in F^*$. If not stated otherwise nets are supposed to be connected.
\noindent
For a node $x$ the sets of nodes
$$^\bullet x:=\{y \in N:\ (y,x)\in F\} \text{ and } x^\bullet :=\{y \in N:\ (x,y)\in F\}$$
denote respectively the \emph{pre-set} and the \emph{post-set} of $x$. The concept generalizes to the pre-set and post-set of sets of nodes. The net $N$ is a
\emph{$T$-net} if all places $p\in P$ satisfy $card\ p^\bullet = card\ ^\bullet p=1$. The net is a \emph{$P$-net} if all transitions $t\in T$ \mbox{satisfy $card\ t^\bullet = card\ ^\bullet t=1$}. The net $N$ is a \emph{free-choice net} if for each \mbox{pair $(p,t) \in P \times{}T$}
$$(p,t)\in F \implies ^\bullet t \times{} p^\bullet \subset F$$
The \emph{cluster} of a node $x$ is the smallest subnet $cl \subset N$ which contains $x$ and for each place $p \in cl$ also its post-set $p^\bullet$ and for each transition $t \in cl$ also its pre-set $^\bullet t$. For a free-choice net $N$ and a cluster $cl \subset N$ holds: Each pair $(p,t)\in cl_P\times{} cl_T$ \mbox{satisfies $(x,y)\in F$}.\bigskip

A \emph{$P$-component} of $N$ is a non-empty, strongly connected $P$-subnet $C \subset N$ such that for each place $p\in C$ holds $^\bullet p \cup p^\bullet \subset C$.
\noindent
\item \emph{Structure and dynamics}: A \emph{marking} of $N$ is a map
$$\mu:N_P \xrightarrow{} \mathbb{N}$$
The \emph{token count} at $\mu$ of a subset $X$ of nodes of $N$ is the number
$$\Vert \mu \Vert_X:=\sum_{p \in X\cap N_P} \mu(p)$$
adding up all tokens marking places of $X$. A \emph{Petri net} or \emph{marked net} is a pair $(N,\mu)$ with $\mu$ a marking of $N$. \mbox{A \emph{$T$-system}} respectively a \emph{free-choice system} is a Petri \mbox{net $(N,\mu)$} with $N$ a $T$-net respectively a free-choice net.

\medskip\noindent
A transition $\,t \in T\,$ is $\,$\emph{enabled}$\,$ at the marking $\,\mu\,$ if all$\,$ its $\,$pre-places$\,$ are$\,$ marked, i.e. if for all  $p\in\ ^\bullet t$ holds $\mu(p) \geq 1$. A transition $t$, which is enabled at $\mu$, may fire. \emph{Firing} $t$ consumes one token from each pre-place of $t$ and creates one token at each post-place of $t$. The notation
$$\mu \xrightarrow{t} \mu_{post}$$
means: $t$ is enabled at $\mu$, and firing $t$ at $\mu$ creates the marking $\mu_{post}$ defined for each $p \in N_P$ as
$$\mu_{post}(p):=\begin{cases}
\mu(p)	 & \text{ if } p \in\ ^\bullet t \cap t^\bullet \text{ or } p \not\in\ ^\bullet t \cup t^\bullet \\
\mu(p)-1 & \text{ if } p \in\ ^\bullet t \setminus t^\bullet \\
\mu(p)+1& \text{ if } p \in\ t^\bullet \setminus ^\bullet t \\
\end{cases}$$
The set of all transitions of $N$ enabled at the marking $\mu$ is denoted $en(N,\mu)$.\bigskip

An \emph{occurrence sequence} $\sigma=(t_1,...,t_n)$ is a sequence of transitions. It is \emph{enabled} at a \mbox{marking $\mu$} if
$$\mu \xrightarrow{t_1} \mu_1 \xrightarrow{t_2} ... \mu_{n-1} \xrightarrow{t_n} \mu_n$$
The shorthand
$$\mu \xrightarrow{\sigma} \mu_n$$
expresses the successive enabledness and firing of the component transitions of $\sigma$.\medskip

A marking $\mu_{post}$ of a net $\,N\,$ is \emph{reachable from a marking} $\,\mu_{pre}\,$ if there exists an occurrence sequence
$$\mu_{pre} \xrightarrow{\sigma} \mu_{post}$$
A marking $\,\mu\,$ is reachable in a Petri net $\,(N,\mu_0)\,$ if $\,\mu\,$ is reachable from $\mu_0$. A marking which
is~reachable from each reachable marking is a \emph{home marking}. A Petri net $(N,\mu_0)$ is \emph{live} if for
each transition $t$ and from each reachable mark\-ing a marking is reachable which enables $t$. A Petri$\,$ net $\,(N,\mu_0)\,$
is $\,$\emph{bounded}$\,$ if there$\,$ exists a constant $\,K\,$ with $\,\mu(p) \leq K\,$ for$\,$ each$\,$ reachable  marking $\mu$
and for all places {$p\in N_P$}.\! The Petri net is \emph{safe} if the bound $K\!=\!1$ is possible.
A~net~{$N$} is \emph{well-formed} if there exists a marking $\mu_0$ of $N$ such that the Petri net $(N,\mu_0)$
is live and bounded. Two reachable markings $\mu_j,\ j=1,2,$ of a live and bounded Petri net with
$\mu_1 \geq \mu_2$ are equal. We will often use the latter result without further mentioning.\medskip

Each well-formed free-choice net is covered by $P$-components. The token count of  a $P$-compon\-ent  is the same for all reachable markings of $(N,\mu_0)$. Each strongly connected $P$-subnet of a well-formed free-choice net, in particular each elementary circuit, is contained in a $P$-component, \cite[Chap. 5]{TV1984}.

\item \emph{Greedy cluster}: For a \mbox{cluster $cl \subset N$} we denote by $\mu_{cl}$ the marking of $N$
$$\mu_{cl}:N_P \xrightarrow{} \mathbb{N},\ \mu_{cl}(p):=\begin{cases}
1 & p \in cl_P \\
0 & \text{otherwise}
\end{cases}$$
Hence $\mu_{cl}$ is the characteristic function of the set of places of $cl$. If a Petri \mbox{net $(N,\mu_0)$} has a reachable marking with $\mu=\mu_{cl}$ for a given \mbox{cluster $cl \subset N$}, then $cl$ shows a kind of ``greediness'' - in particular if $\mu$ is a home marking. Greediness will be a fundamental property in this paper.
\end{enumerate}
\end{remark}

%-------------------------------------------------------------------------------------------------------------------------------------------------------------------
\begin{flushleft}
We emphasize the following properties of $T$-systems:
\end{flushleft}
\begin{remark}[$T$-nets and $T$-systems]\label{rem_enablingtsystem}
\begin{itemize}
\item The $P$-components of a $T$-net are its elementary circuits.

\item In a $T$-system an enabled transition can lose its enabledness only by firing itself. During the firing of an occurrence sequence the token count of a path
$(p_{in},...,p_{out})$ joining two places can change only by creating tokens at $p_{in}$ or consuming tokens at $p_{out}$. The token count of a circuit is the same for each reachable marking.
\end{itemize}
\end{remark}

%======================================================================================
\section{$CP$-subnets and $CP$-exhaustion}\label{sect_cpexhaustion}

%-------------------------------------------------------------------------------------------------------------------------------------------------------------------
$CP$-subnets have been introduced by Desel and Esparza. Since then, $CP$-subnets are a standard tool for the investigation of free-choice systems, \cite[Def. 7.7, Theor. 7.13]{DE1995}.
\begin{definition}[$CP$-subnet]\label{def_cpsubnet}
Consider a net $N$.
\begin{enumerate}
\item A non-empty, weakly connected transition-bordered $T$-subnet
$$\hat{N} \subset N$$
is a \emph{$CP$-subnet} of $N$ if the complement
$$\overline{N}:=N \setminus \hat{N}$$
contains some transition and is strongly connected.
\item A \emph{way-in transition} $t_{in}$ of a $CP$-subnet $\hat{N} \subset N$ is a transition $t \in \hat{N}_T$ with
        $^\bullet t \cap \overline{N}_P \neq \emptyset$,
        a    \emph{way-out transition} $t_{out}$ is a transition $t \in \hat{N}_T$ with $t^\bullet \cap \overline{N}_P \neq \emptyset$.
   \eject
\item A $CP$-subnet $\hat{N} \subset N$ is \emph{adapted} to a cluster $cl \subset N$ if $cl \not\subset \hat{N}$.
\end{enumerate}
\end{definition}

%-------------------------------------------------------------------------------------------------------------------------------------------------------------------
\begin{flushleft}
We recall some well-known properties of $CP$-subnets.
\end{flushleft}
\begin{remark}[Existence, structure and dynamics of $CP$-subnets]\label{rem_existencestructuredynamics}
Consider a well-formed free-choice net $N$.
\begin{enumerate}
\item \emph{Existence}: If $N$ is not a $T$-net then $N$ has a $CP$-subnet $\hat{N} \subset N$, \mbox{\cite[Prop. 7.11]{DE1995}}. One may even assume that a given transition $t\in N$ is not contained \mbox{in $\hat{N}$}, \cite[Lem. 1.2]{Wehler2010}. The latter result is crucial to obtain in Theorem \ref{theor_existencecpexhaustion} a $CP$-exhaustion of $N$ which is adapted to a given regeneration cluster of $N$.

\item \emph{Structure}: A $CP$-subnet $\hat{N} \subset N$ has the following structural properties:

\begin{itemize}
\item The net $\hat{N}$ has a unique way-in transition $t_{in} \in \hat{N}_T$, \mbox{\cite[Prop. 7.10]{DE1995}}. The net $\hat{N}$ has at least one way-out transition $t_{out}$ because $N$ is strongly connected.\vspace*{0.8mm}

\item Each place $p \in \hat{N}$ has a path $(t_{in},...,p) \subset \hat{N}$ leading from $t_{in}$ to $p$, cf. \mbox{\cite[Prop. 7.10 proof]{DE1995}}.\vspace*{0.8mm}

\item The complement $\overline{N}:=N\setminus \hat{N}$ is a well-formed free-choice net too, \cite[Cor. 7.9]{DE1995}.\vspace*{0.8mm}

\item Each cluster $cl \subset N$ satisfies
$$cl \not\subset \hat{N} \iff (cl \cap \overline{N})_T \neq \emptyset \iff cl_P \subset \overline{N}_P$$
The proof uses that $\overline{N}$ is strongly connected and that $\hat{N} \subset N$ is transition-bordered.
\end{itemize}

\item \emph{Dynamics}: Consider a $CP$-subnet $\hat{N} \subset N$, and a live and bounded marking $\mu_0$ of $N$ and a reachable marking $\mu$ of $(N,\mu_0)$.

\begin{itemize}
\item There exists a \emph{shutdown sequence} for $\hat{N}$, i.e. a finite occurrence sequence \mbox{of $(N,\mu)$}
$$\mu \xrightarrow{\sigma_{sd}} \mu_{sd}$$
with transitions $t\in \sigma_{sd}$ only from $\hat{N}_T\setminus \{t_{in}\}$, such that $\mu_{sd}$ enables no transition of $\hat{N}$ different from $t_{in}$, \cite[Prop. 7.8]{DE1995}.

\item The free-choice system
$$(\overline{N},\overline{\mu}_{sd}) \text{ with } \overline{\mu}_{sd}:=\mu_{sd}|\overline{N}$$
is live and bounded, \cite[Prop. 7.8]{DE1995}. If $(N,\mu_0)$ is safe then $(\overline{N},\overline{\mu}_{sd})$ is safe too.

\item An occurrence sequence of transitions from $\overline{N}$ is enabled at a marking $\mu$ \mbox{of $N$} iff it is enabled as an occurrence sequence of  the restriction $(\overline{N},\mu|\overline{N})$.
\end{itemize}
\end{enumerate}
\end{remark}

%-------------------------------------------------------------------------------------------------------------------------------------------------------------------
\begin{flushleft}
Proposition \ref{prop_iterationcpsubnet} shows: A $CP$-subnet in the complement of a $CP$-subnet is \mbox{a $CP$-subnet} in the original net too. The result prepares the induction step in the proof of Theorem \ref{theor_existencecpexhaustion}.
\end{flushleft}
\begin{proposition}[Iteration of $CP$-subnets]\label{prop_iterationcpsubnet}
Consider a well-formed free-choice net $N$ and \mbox{a $CP$-subnet $\hat{N}_0 \subset N$} with complement
$$\overline{N}_0:=N \setminus \hat{N}_0.$$
Each $CP$-subnet $\hat{N}_1 \subset \overline{N}_0$ of the complement $\overline{N}_0$ is also a $CP$-subnet $\hat{N}_1 \subset N \text{ of } N$.
\end{proposition}
\begin{proof}
One has to show that the complement $N \setminus \hat{N}_1$ contains some transition and is strongly connected.
\begin{flushleft}
i) \emph{The complement contains some transition}: Because $\hat{N}_1 \subset \overline{N}_0$ is a $CP$-subnet, the \mbox{complement $\overline{N}_0 \setminus \hat{N}_1$}
contains some transition by definition. The inclusion
\end{flushleft}
$$(\overline{N}_0 \setminus \hat{N}_1) \subset (N \setminus \hat{N}_1)$$
implies that the complement $N \setminus \hat{N}_1$ contains a transition too.
\begin{flushleft}
ii) \emph{The complement is strongly connected}: The proof relies on Remark \ref{rem_existencestructuredynamics}, part 2). Set
$$\overline{N}_1:=\overline{N}_0\setminus \hat{N}_1.$$
By construction $\overline{N}_1$ is strongly connected and
$$N \setminus \hat{N}_1= span_N<\overline{N}_1,\hat{N}_0>$$
The claim that two nodes
$$x_1,\ x_2 \in N \setminus \hat{N}_1$$
can be joined in both directions by a path in $span_N<\overline{N}_1,\hat{N}_0>$ reduces to the following two cases:
\begin{itemize}
\item $x_1 \in \overline{N}_1 \text{ and } x_2 \in \hat{N}_0$: First, there exists a place
$$p_{in} \in\ ^\bullet t_{in} \subset \overline{N}_1$$
with $t_{in} \in\hat{N}_0$ the way-in transition of $\hat{N}_0$. Because
$p_{in} \in \overline{N}_1$ there exists a path
$$\gamma_1:=(x_1,...,p_{in}) \subset \overline{N}_1$$
Secondly, choose
$$\gamma_2:=(p_{in},t_{in}) \subset span_N<\overline{N}_1,\hat{N}_0>$$
the joining edge. Eventually, there exists a path
$$\gamma_3:=(t_{in},...,x_2) \subset \hat{N}_0.$$
The concatenation satisfies
$$\gamma:=\gamma_1 *\gamma_2 *\gamma_3=(x_1,...,x_2) \subset span_N<\overline{N}_1,\hat{N}_0>$$
For the opposite direction: First, there exists a path
$$\delta_1:=(x_2,...,t_{out}) \subset \hat{N}_0$$
with $t_{out} \in \hat{N}_0$ a suitable way-out transition of $\hat{N}_0$. Secondly, choose a place
$$p_{out} \in (t_{out})^\bullet \subset \overline{N}_1$$
and set
$$\delta_2:=(t_{out},p_{out})  \subset span_N<\overline{N}_1,\hat{N}_0>$$
the joining edge. Eventually, there exists a path
$$\delta_3:=(p_{out},...,x_1) \subset \overline{N}_1$$
The concatenation satisfies
$$\delta:=\delta_1*\delta_2*\delta_3=(x_2,...,x_1) \subset span_N<\overline{N}_1,\hat{N}_0>$$
\item Both $x_1, x_2 \in \hat{N}_0$: The case follows from the first case after introducing an intermediate \mbox{place $x_3 \in \overline{N}_1$}.\qed
\end{itemize}
\end{flushleft}

\vspace*{-5mm}
\end{proof}

%-------------------------------------------------------------------------------------------------------------------------------------------------------------------
\begin{flushleft}
The main means for our proof of van der Aalst's theorem is the new concept of \mbox{a $CP$-exhaustion}. Theorem \ref{theor_existencecpexhaustion} shows that any well-formed free-choice net has \mbox{a $CP$-exhaustion} adapted to a given cluster.
\end{flushleft}
\begin{definition}[Adapted $CP$-exhaustion]\label{def_cpexhaustion}
A \emph{$CP$-exhaustion} of a net $N$ is a family
$${(\hat{N}_i)}_{i \in I},\ I=\{0,...,n\} \subset \mathbb{N},$$
of pairwise disjoint $CP$-subnets $\hat{N}_i \subset N$ such that
$$\overline{N}:=N \setminus \dot\bigcup_{i\in I} \hat{N}_i$$
is a strongly connected $T$-net. The $CP$-exhaustion defines the disjoint union
$$N_{exh}:=\overline{N} \  \dot\cup \  \dot\bigcup_{i\in I} \hat{N}_i$$
The $CP$-exhaustion is \emph{adapted} to a given cluster $cl \subset N$ if $\hat{N}_i \subset N$ is $cl$-adapted for \mbox{all $i\in I$}.
\end{definition}

\begin{flushleft}
The net $N_{exh}$ has the same nodes as $N$ and
\end{flushleft}
$$span_N<N_{exh}>=N$$
For a non-empty index set $I$ the net $N_{exh}$ is not connected and $N_{exh} \subset N$ is not a full subnet.

\eject
%-------------------------------------------------------------------------------------------------------------------------------------------------------------------
\medskip\noindent
In general one does not obtain a $CP$-exhaustion by just taking a maximal family of pairwise disjoint $CP$-subnets of $N$:
Their complement $\overline{N}$ is not necessarily connected. Therefore we construct a $CP$-exhaustion of a well-formed free-choice net iteratively: The next $CP$-subnet is a $CP$-subnet of the \emph{complement} of the previous $CP$-subnet. This is a stronger property than just being a $CP$-subnet
of $N$. Proposition \ref{prop_iterationcpsubnet} ensures that all obtained $CP$-subnets are also $CP$-subnets of $N$.

\begin{theorem}[Existence of an adapted $CP$-exhaustion]\label{theor_existencecpexhaustion}
Consider a well-formed free-choice net $N$ and a cluster $cl \subset N$. Then $N$ has a $cl$-adapted $CP$-exhaustion.
\end{theorem}
\begin{proof}
i) \emph{Algorithm constructing the exhaustion}: The following algorithm constructs the index \mbox{set $I \subset \mathbb{N}$} and the family ${(\hat{N}_i)}_{i \in I}$ of the $CP$-exhaustion by induction on $j \geq 0$:
\begin{flushleft}
Initialize the net $\overline{N}_{-1}:=N$ and the cluster $cl_{-1}:=cl \subset \overline{N}_{-1}$.
\end{flushleft}
\begin{flushleft}
Step $\mathfrak{A}(j)$ constructs a triple $(\hat{N}_j,\ \overline{N}_j,\ cl_j)$ with the following properties
\end{flushleft}
\begin{itemize}
\item The first component $\hat{N}_{j}$ is a $cl_{j-1}$-adapted $CP$-subnet $\hat{N}_{j} \subset \overline{N}_{j-1}$.
\item Second component: The complement $\overline{N}_{j}:=\overline{N}_{j-1} \setminus \hat{N}_j$ is well-formed.
\item Third component: The cluster $cl_j:=cl_{j-1} \cap \overline{N}_j$ is not empty .
\end{itemize}
\begin{flushleft}
Induction start $\mathfrak{A}(0)$: If $N$ is a $T$-net then set
\end{flushleft}
$$I:=\emptyset,\ \overline{N}:=N$$
and terminate. Otherwise Remark \ref{rem_existencestructuredynamics} provides a $cl$-adapted $CP$-subnet
$$\hat{N}_0 \subset N$$
Define
$$\overline{N}_0:=N \setminus \hat{N}_0 \text{ and } cl_{0}:=cl \cap \overline{N}_0 \neq \emptyset$$
Remark \ref{rem_existencestructuredynamics}, part 2) shows that $\overline{N}_0$ is well-formed.
\begin{flushleft}
Induction step $j \mapsto j+1$: By induction assumption $\mathfrak{A}(j)$ the free-choice net $\overline{N}_{j}$ is well-formed. If $\overline{N}_{j}$ is a $T$-net, then set
$$I:=\{0,...,j-1\} \text{ and } \overline{N}:=\overline{N}_j$$
and terminate. Otherwise Remark \ref{rem_existencestructuredynamics} provides a $cl_j$-adapted $CP$-subnet
$$\hat{N}_{j+1} \subset \overline{N}_j$$
Define
$$\overline{N}_{j+1}:=\overline{N}_j\setminus \hat{N}_{j+1} \text{ and } cl_{j+1}:=cl_j \cap \overline{N}_{j+1} \neq \emptyset$$
\end{flushleft}

\noindent
ii) \emph{Correctness}: The algorithm terminates because $N$ is finite. If $N$ is not a $T$-net then the iterative application of Proposition \ref{prop_iterationcpsubnet} implies that for each $j\in I$ the $CP$-subnet
$$\hat{N}_{j} \subset \overline{N}_{j-1}$$
is also a $CP$-subnet $\hat{N}_{j} \subset N$. Also $cl  \not\subset \hat{N}_j$ because by construction
$$cl \cap \overline{N}_{j}=cl_{j-1} \cap \overline{N}_{j} \neq \emptyset$$

\vspace*{-9mm}
\end{proof}

\begin{example}[Adapted $CP$-exhaustion]\label{exam_cpexhaustion}
The example applies the $CP$-algorithm from \mbox{Theorem \ref{theor_existencecpexhaustion}} to the free-choice net $N$ underlying the Petri net $(N,\mu_0)$ from Figure \ref{fig_figure5}. The net is taken from \cite{Aal2018}. It is well-formed because $\mu_0$ is live and safe. We construct by iteration \mbox{a $CP$-exhaustion} \mbox{of $N$} adapted to the cluster
$$cl:=span_N<start,\ t_0>.$$
\begin{figure}[!h]
\vspace*{-3mm}
\centering
\includegraphics[scale=0.4]{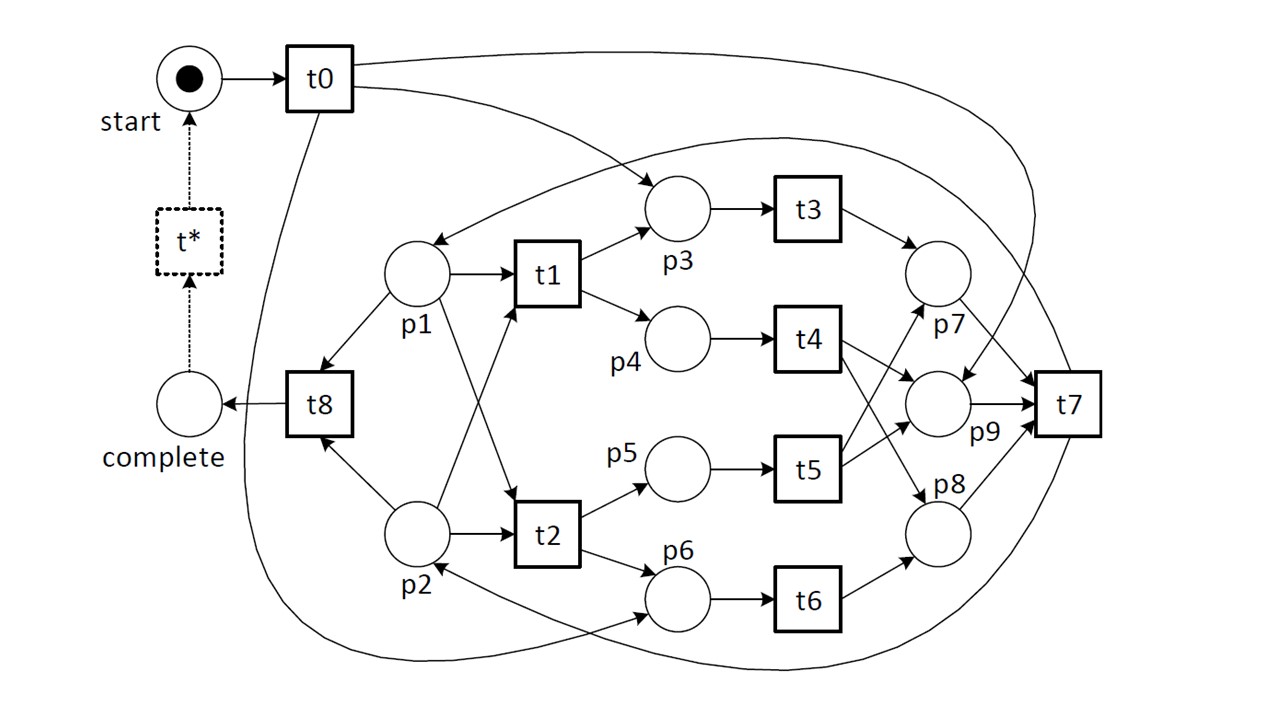}\vspace*{-3mm}
\caption{Free-choice system $(N,\mu_0)$ from \cite[Fig. 5]{Aal2018} (dashing at $t^{*}$ here not significant)}
\label{fig_figure5}\vspace*{-1mm}
\end{figure}
\begin{enumerate}
\item \emph{Constructing an adapted $CP$-exhaustion}:
\begin{itemize}
\item First, choose the $CP$-subnet of $N$
$$\hat{N}_0:=span_N<p_4,t_1,t_4>\ \subset N$$
The complement
$$\overline{N}_0:=N \setminus \hat{N}_0$$
is well-formed.
\eject
\item Secondly, choose the $CP$-subnet of the complement $\overline{N}_0$
$$\hat{N}_1:= span_{\overline{N}_0}<p_5,t_2,t_5>\ \subset \overline{N}_0$$
\item The final complement is the strongly connected $T$-net
$$\overline{N}:=\overline{N}_0 \setminus \hat{N}_1$$
\end{itemize}

\noindent Figure \ref{fig_cpexhaustion} shows the subnets $\hat{N}_0,\ \hat{N}_1,\ \overline{N} \subset N$. The family $(\hat{N}_0,\hat{N}_1)$ is a $cl$-adapted $CP$-exhaustion of $N$, and
$$N_{exh}:=\overline{N}\ \dot \cup \ \hat{N}_0\ \dot\cup \ \hat{N}_1$$
satisfies
$$span_N<N_{exh}>=N.$$

\noindent Note: The $CP$-exhaustion is also adapted to the cluster
$$cl_1=span_N<p_1,p_2,t_1,t_2,t_8>.$$
\item \emph{Greediness of the clusters}: Both clusters $cl$ and $cl_1$ of $N$ are greedy in the Petri \mbox{net $(N,\mu_0)$}. They provide examples of \emph{regeneration clusters}, a fundamental concept which will be introduced in Definition \ref{def_perpetual}.
\end{enumerate}
\end{example}

\begin{figure}[H]
\centering
\includegraphics[scale=0.42]{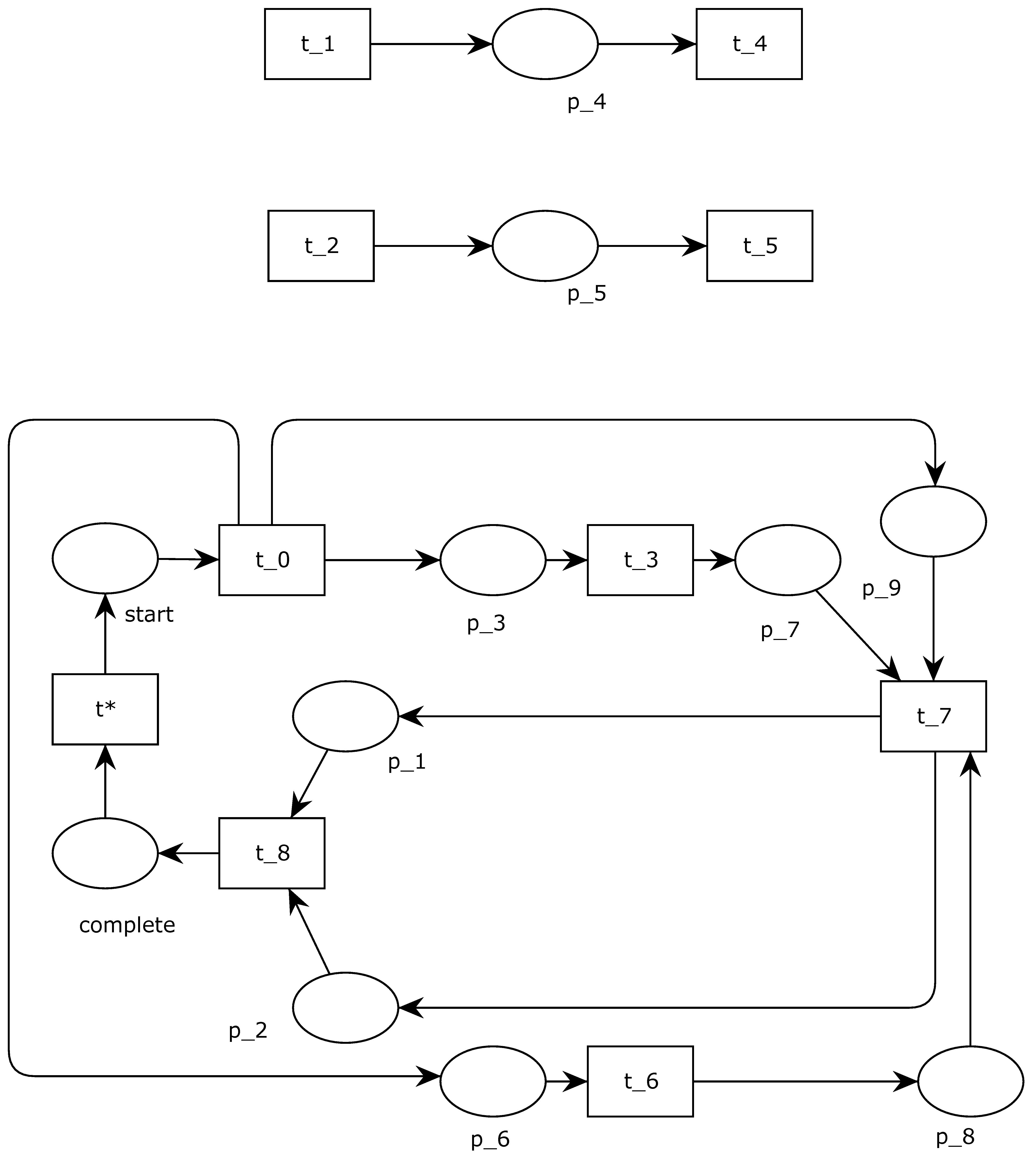}
\caption{$CP$-exhaustion of the net $N$ from Figure \ref{fig_figure5}. Top: $\hat{N}_0,\hat{N}_1$; bottom: $\overline{N}$}
\label{fig_cpexhaustion}
\end{figure}

%-------------------------------------------------------------------------------------------------------------------------------------------------------------------
\begin{definition}[Way-in places and critical transitions]\label{def_criticalboundary}
Consider a $CP$-exhaustion ${(\hat{N}_i)}_{i\in I}$ of a net $N$. Set
$$\hat{N}:= \bigcup_{i\in I}\ \hat{N}_i \text{ and  } \overline{N}:=N \setminus \hat{N}$$
The elements from
$$\overline{N}_{P,in}:=\overline{N}_P \cap (\hat{N}_T)^\bullet \subset \overline{N}_P$$
are the \emph{way-in places} of $\overline{N} \subset N$. Their post-transitions, the elements from
$$\overline{N}_{T,in}:= (\overline{N}_{P,in}) ^\bullet \subset \overline{N}_T$$
are the \emph{critical transitions} of $\overline{N}$. Here the post-place operator in $(\hat{N}_T)^\bullet$ applies with respect to the edges \mbox{of $N$}, while the post-transition operator in $(\overline{N}_{P,in}) ^\bullet$ applies with respect to the edges of $\overline{N}$.
\end{definition}
\begin{flushleft}
For an application of the concepts of Definition \ref{def_criticalboundary} see Figure \ref{fig_enablingcriticaltransitions} with two way-in \mbox{places $p_j$} and two critical transitions $t_j,\ j=1,2$.
\end{flushleft}

%======================================================================================
\section{Enabling equivalence and marking equality in free-choice systems}\label{sect_enablingequaivalencemarkingequality}

%-------------------------------------------------------------------------------------------------------------------------------------------------------------------
Van der Aalst introduces the two fundamental concepts from Definition \ref{def_perpetual}.
\begin{definition}[Lucency and perpetual Petri net]\label{def_perpetual}
\begin{enumerate}
\item A Petri net $(N,\mu_0)$ is \emph{lucent} if for any pair $(\mu_1,\mu_2)$ of reachable markings
$$en(N,\mu_1)=en(N,\mu_2) \implies \mu_1=\mu_2.$$
\item A Petri net $(N,\mu_0)$ is \emph{perpetual} if it is live and bounded and there exists a \mbox{cluster $cl \subset N$} such that $\mu_{cl}$ is a home marking of $(N,\mu_0)$. The cluster $cl$ is named a \emph{regeneration cluster} of $(N,\mu_0)$, and $\mu_{cl}$ is a \emph{regeneration marking} of $(N,\mu_0)$.
\end{enumerate}
\end{definition}

\noindent In \cite{Aal2018} the cluster $cl$ is named a \emph{home cluster} and paraphrased as a ``regeneration point''. Different than~\cite{Aal2018} we prefer the name \emph{regeneration cluster}. The term \emph{home cluster} could suggest erroneously that any home marking relates to a home cluster. The property to be a regeneration cluster depends on the Petri net $(N,\mu_0)$, not alone on the subnet $cl \subset N$.\medskip

We will often rely on the \emph{fundamental property} of reachable markings $\mu$ in a perpetual free-choice system $(N,\mu_0)$ with regeneration cluster $cl$:
\begin{itemize}
\item Each $P$-component $C \subset N$ contains exactly one place of $cl$ and has token \mbox{count $\Vert \mu \Vert_C=1$}.
\item The Petri net $(N,\mu_0)$ is safe.
\item If $N$ is a $T$-net then each elementary circuit $\gamma \subset N$ contains the unique \mbox{transition $t_{cl} \in cl_T$} and has token count $\Vert \mu \Vert_\gamma =1$.
\end{itemize}

Proof:
 Each $P$-component $C\subset N$ has a positive token count at $\mu_{cl}$. \mbox{Hence $C$} contains exactly one place of $cl$ and satisfies
$\Vert \mu_{cl} \Vert_C=1$. The token count \mbox{of $C$} is the same for all reachable markings. Because the well-formed free-choice net $N$ is covered \mbox{by $P$-components} each perpetual free-choice system $(N,\mu_0)$ is safe. In the particular case of a $T$-net the $P$-components are exactly the elementary circuits of $N$.

\medskip
We consider the whole subject of lucency as a question about two equivalence relations on the set of reachable markings of a given Petri net $(N,\mu_0)$: In addition to the equality of reachable markings one considers the relation of \emph{enabling equivalence}.
%\end{proof}

%-------------------------------------------------------------------------------------------------------------------------------------------------------------------
\begin{definition}[Enabling equivalence]\label{def_enablingequivalence}
A pair of reachable markings $(\mu_1,\mu_2)$ of a Petri net $(N,\mu_0)$ is \emph{enabling equivalent} if
$$en(N,\mu_1)=en(N,\mu_2)$$
\end{definition}
%\begin{flushleft}
\noindent Then the Petri net problem under consideration reads: When does enabling equivalence imply marking equality?
%\end{flushleft}

\medskip\noindent Our proof of van der Aalst's theorem, see Theorem \ref{theor_vanaalsttheorem}, starts with a \mbox{pair $(\mu_1,\mu_2)$} of enabling equi\-valent markings of $N$. The well-formed free-choice net $N$ has \mbox{a $cl$-adapted $CP$-decomposition} with a final strongly connected $T$-net $\overline{N}$. The proof relies on firing a global shutdown sequence $\sigma$, the concatenation of shutdown sequences for all $CP$-subnets. The firing squeezes out all tokens from \mbox{the $CP$-subnets} and creates a pair $(\overline{\mu}_{1,sd},\overline{\mu}_{2,sd})$ of markings of $\overline{N}$. These markings are still enabling equivalent with respect to the resulting marking of $\overline{N}$. Hence the original claim reduces to the analogous claim for a perpetual marking of $\overline{N}$. Here the marking equality follows by elementary methods \mbox{for $T$-systems}. During the proof we have to keep an eye on how the $CP$-algorithm from \mbox{Theorem \ref{theor_existencecpexhaustion}} propagates in each step the following properties
$$\text{ well-formedness, perpetuality, enabling equivalence, and marking equality}.$$
The logical dependencies between the intermediate results is clarified by the diagram from Figure \ref{fig_planofproof}:
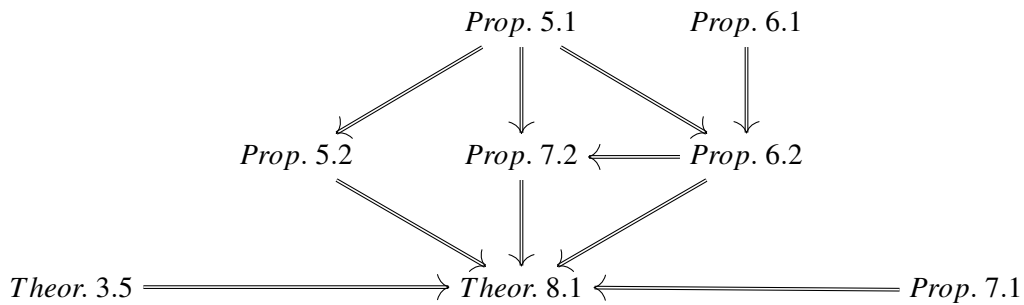
\begin{figure}[H]
$$\begin{tikzpicture}
\matrix (m) [matrix of math nodes,row sep=3em,column sep=3em,minimum width=3em] {
 	& 	& Prop.\ \ref{prop_perpetualtsystem} 	& Prop.\ \ref{prop_cpnetsubnetsperpetual}	&	 \\
	& Prop.\  \ref{prop_tsystemlucent}		& Prop.\  \ref{prop_propagatingenablingequivalence}	&  Prop.\  \ref{prop_markingequality}	&   \\	
Theor.\  \ref{theor_existencecpexhaustion}&  		&	Theor.\  \ref{theor_vanaalsttheorem}	&	& Prop.\  \ref{prop_propagatingperpetuality}	\\};
\path[-stealth] ;
    \draw[->, double] (m-3-1) --  (m-3-3);
    \draw[->, double] (m-1-3) -- 	(m-2-2);
    \draw[->, double] (m-1-3) -- 	(m-2-4);
    \draw[->, double] (m-2-4) -- 	(m-2-3);
    \draw[->, double] (m-2-2) -- 	(m-3-3);
    \draw[->, double] (m-2-4) -- 	(m-3-3);	
    \draw[->, double] (m-1-3) -- 	(m-2-3);
    \draw[->, double] (m-1-4) --  (m-2-4);
    \draw[->, double] (m-2-3) --  (m-3-3);
    \draw[->, double] (m-3-5) --  (m-3-3);
\end{tikzpicture}$$\vspace*{-10mm}
\caption{Overview of the proof of Theorem \ref{theor_vanaalsttheorem}}
\label{fig_planofproof}
\end{figure}

\begin{itemize}
\item Theorem \ref{theor_existencecpexhaustion}	splits the net of a perpetual free-choice system into an adapted $CP$-exhaustion with a final
strongly connected $T$-net.

\item Propositions \ref{prop_perpetualtsystem} and \ref{prop_tsystemlucent} study for $T$-systems the interplay of their deterministic occurrence semantics with a regeneration cluster. Proposition \ref{prop_perpetualtsystem} collects relevant properties of  perpetual $T$-systems.
    Proposition~\ref{prop_tsystemlucent} concludes that perpetual $T$-systems are lucent.

\item The analogue for adapted $CP$-nets in perpetual free-choice systems is proved in Proposition \ref{prop_cpnetsubnetsperpetual} and \ref{prop_markingequality}.

\item Proposition \ref{prop_propagatingperpetuality} and \ref{prop_propagatingenablingequivalence} ensure: Each step of the $CP$-exhaustion algorithm from
Theorem \ref{theor_existencecpexhaustion} propagates perpetuality and enabling equivalence to the next level.

\item Theorem \ref{theor_vanaalsttheorem} restates and proves van der Aalst's theorem.
\end{itemize}
\noindent
The intermediate results will be proved in Sections \ref{sect_tsystemsperpetualitymarkings}, \ref{sect_cpsystemsperpetualitymarkings} and \ref{sect_propagating}. Section \ref{sect_proofvanderaalsttheorem} brings together all results to show van der Aalst's theorem.

%======================================================================================
\section{Enabling equivalence and marking equality in perpetual $T$-systems}\label{sect_tsystemsperpetualitymarkings}

The present section proves van der Aalst's theorem in the particular case of a \mbox{perpetual $T$-system}, see Theorem \ref{prop_tsystemlucent}. The proof for $T$-systems is much easier than the proof for free-choice systems in general. In the presence of a regeneration cluster marking equivalence provides certain distinguished paths of the underlying $T$-net. Due to Proposition \ref{prop_perpetualtsystem}, part 2 ii) these paths are safe in the perpetual $T$-system.

%-------------------------------------------------------------------------------------------------------------------------------------------------------------------
\begin{proposition}[Token count of paths in $T$-systems]\label{prop_perpetualtsystem}
Let $(TN,\mu)$ be a $T$-system. For each transition $t\in TN_T$ and pre-place $q \in\ ^\bullet t$ \mbox{with $\mu(q)=0$} denote by
$$en(TN,\mu)_q:=\left\{(\tau,\delta):\ \tau \in en(TN,\mu),\ \delta=(\tau,...,q,t) \subset TN \text{ elementary}\right\}$$
the set of all enabled transitions $\tau$ together with their elementary paths to $t$, which \mbox{pass $q$}.
\begin{enumerate}
\item \emph{General $T$-system}: Consider an arbitrary \mbox{transition $t\in TN_T$} which is enabled at $\mu$. For each pre-place $q \in\ ^\bullet t \text{ with } \mu(q)=0$ exists a pair
$$(\tau,\delta) \in en(TN,\mu)_q \text{ with } \Vert \mu \Vert_\delta =0.$$
\item \emph{Perpetual $T$-system}: Assume that $(TN,\mu)$ is even perpetual with regeneration cluster $cl$ and denote \mbox{by $t_{cl}$} the unique transition of $cl$.
\begin{flushleft}
i) For each transition $t\in TN_T$ with a pre-place $p\in \ ^\bullet t$ with $\mu(p)=1$ exists for each \mbox{pre-place $q \in\ ^\bullet t$} with $\mu(q)=0$ a pair
$$(\tau,\delta) \in en(TN,\mu)_q \text{ with } \Vert \mu \Vert_\delta =0 \text{ and } t_{cl} \notin \delta_{seg}$$
for the segment $\delta_{seg}:=(\tau,...,q)$ of $\delta$.

ii) Each elementary path $\delta \subset TN$ with $t_{cl} \notin \delta$ has token count $\Vert \mu \Vert_\delta \leq 1$.
\end{flushleft}
\end{enumerate}
\end{proposition}

\begin{proof}

\vspace*{-7mm}
\begin{enumerate}
\item Because $(TN,\mu)$ is a $T$-system, Remark \ref{rem_enablingtsystem} implies $en(TN,\mu)_q \neq \emptyset$. The following algorithm returns a solution $(\tau,\delta)$:

\noindent Initialization: Define the pair $(q^\prime,\delta_{tail}):=(q,\delta_{t})$ with $\delta_{t}:=(t)$ the constant path.

Iteration step $(q^\prime,\delta_{tail})$: Save $\delta_{old}:=\delta_{tail}$. Because $\Vert \mu \Vert_{\delta_{tail}}=0$ also the transition
$$t^\prime \in TN_T \text{ with } \{t^\prime\} = {(q^\prime)} ^\bullet$$
is enabled at a reachable marking of $(N,\mu)$. Due to $\mu(q^\prime)=0$ there exists a pair
$$(\tau^\prime,\delta^\prime) \in en(TN,\mu)_{q^\prime}$$
\begin{itemize}
\item If $\Vert \mu \Vert_{\delta^\prime}=0$ then return $(\tau,\delta):=(\tau^\prime,\delta^\prime * \delta_{old})$.

\item Otherwise choose the uniquely determined transition $t_{split} \in \delta^\prime$ such that the tail \mbox{of $\delta^\prime$}
$$\delta_{tail}:=(t_{split},...,q^\prime,t^\prime) \text{ satisfies } \Vert \mu \Vert_{\delta_{tail}}=0,$$
and set
$$\delta_{new}=\delta_{tail}*\delta_{old}.$$
If $t_{split} \in en(TN,\mu)$ then return $(\tau,\delta):=(t_{split},\delta_{new})$.

\medskip Otherwise choose a pre-place $q_{pre} \in\ ^\bullet t_{split}$ with $\mu(q_{pre})=0$ and reiterate with
$$(q^\prime,\delta_{tail}):=(q_{pre},\delta_{new})$$
\end{itemize}

\noindent The iteration terminates after finitely many steps: The length of the token-free tail increases during each step. But the length is bounded because the iteration does not construct a token-free circuit. Figure \ref{fig_activation} illustrates the first iteration step.

\begin{figure}[!ht]
%\vspace*{-1mm}
\centering
\includegraphics[width=7.2cm]{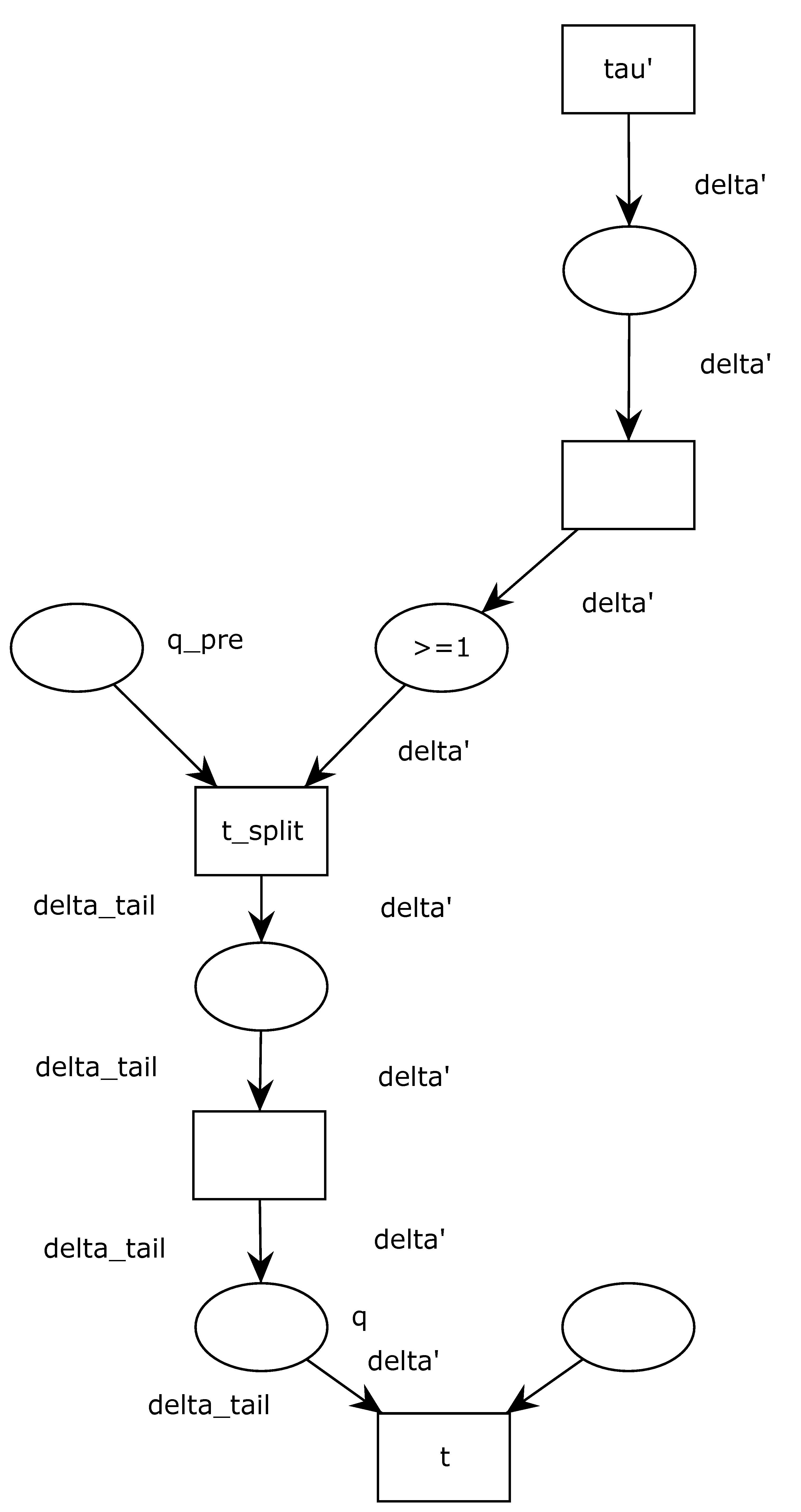}
\caption{First iteration step, case $t_{split} \notin en(TN,\mu)$}
\label{fig_activation}
\end{figure}

\item
i) Assume $\mu(p)=1$. Part 1) provides a pair
$$(\tau, \delta) \in en(TN,\mu)_q \text{ with } \Vert \mu \Vert_{\delta_q}=0$$
The firing of a minimal occurrence sequence $\mu \xrightarrow{\sigma} \mu_{post}$ with $t \in en(TN,\mu_{post})$ forwards all tokens on the pre-places of $\tau$ along $\delta$ \mbox{to $q$}. Minimality of $\sigma$ ensures that the token at $p$ is frozen during the firing \mbox{of $\sigma$}. Then the greediness of $cl$ \mbox{implies $t_{cl} \notin \sigma$} and a posteriori $t_{cl} \notin \delta_{seg}$.

\noindent ii) The proof is indirect. W.l.o.g.
$$\delta=(p_1,...,p_2) \subset TN,\ p_1\neq p_2,$$
is elementary with $\mu(p_1) = \mu(p_2) = 1$. For $j=1,2$ there exist two elementary circuits
$$\gamma_j \subset TN \text{ with } p_j \in \gamma_j.$$
Due to the fundamental property of the perpetual $T$-system $(TN,\mu)$ both circuits have token count
$$\Vert \mu_{cl} \Vert_{\gamma_j}=1 \text{ with } t_{cl} \in \gamma_1 \cap \gamma_2$$
Decompose each $\gamma_j$ as the concatenation
$$\gamma_j=\gamma_{j1}*\gamma_{j2}$$
with the segments
$$\gamma_{j1}=(t_{cl},...,p_j) \text{ and } \gamma_{j2}:=(p_j,...,t_{cl})$$
Claim: The concatenation
$$\gamma_{11}* \delta *\gamma_{22}$$
induces a circuit $\gamma$ which is elementary. Otherwise there exist a node
$$x_1 \in \gamma_{11} \cap \delta \text{ or } x_2 \in \delta \cap \gamma_{22} \text{ or } x_3 \in \gamma_{11} \cap \gamma_{22}$$
In case of a node
$$x_1 \in \gamma_{11} \cap \delta$$
the concatenation of the segments
$$(p_1,...,x_1) \text{ of } \delta \text{ and } (x_1,...,p_1) \text{ of } \gamma_{11}$$
induces a circuit which avoids $t_{cl}$. Analogously, in case of a node
$$x_2 \in \delta \cap \gamma_{22}$$
the concatenation of the segments
$$(p_2,...,x_2) \text{ of } \gamma_{22} \text{ and } (x_2,...,p_2) \text{ of } \delta$$

\begin{figure}[H]
\centering
\includegraphics[width=10.2cm]{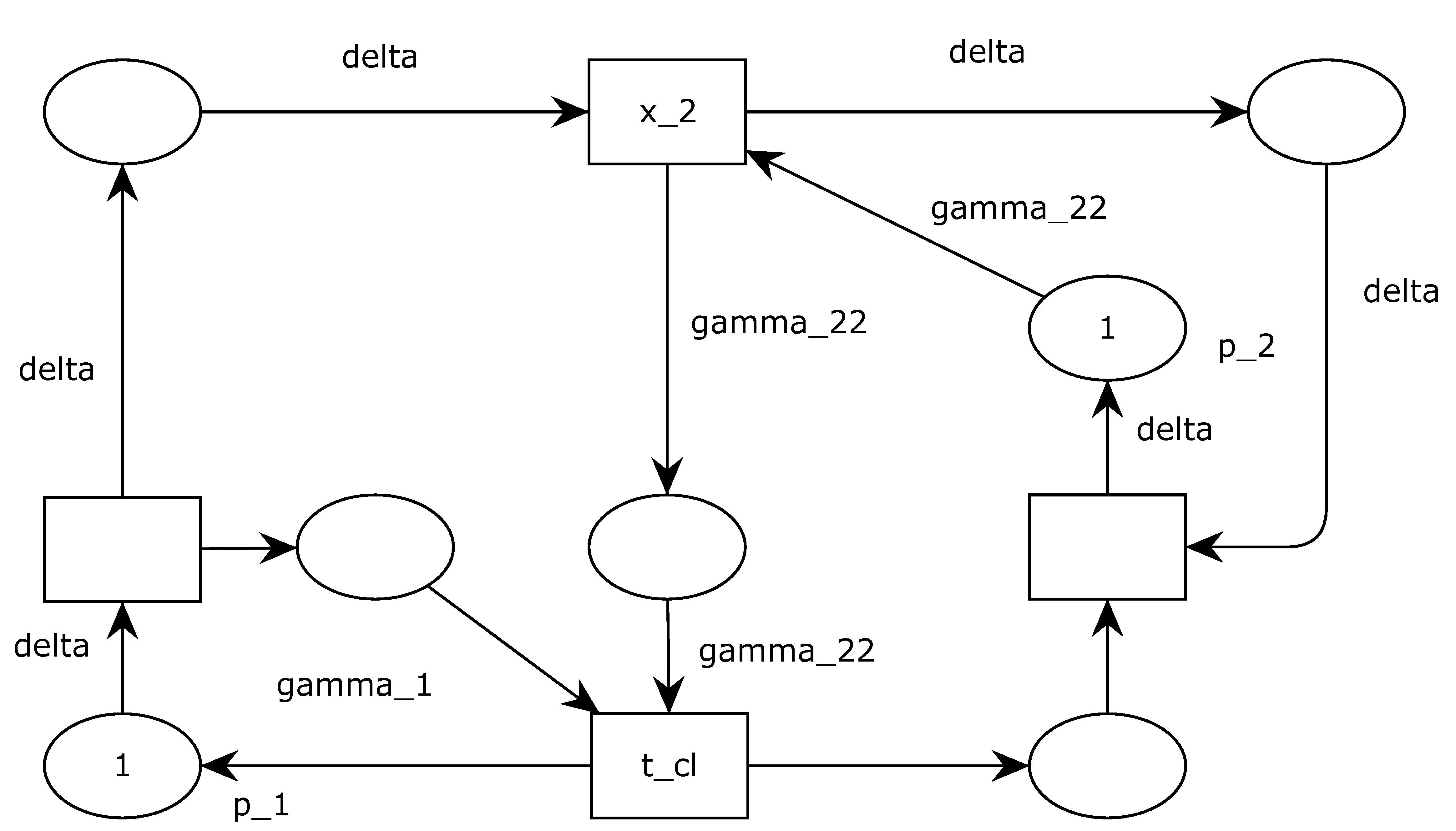}
\caption{Indirect proof: Common node $x_2 \in \delta \cap \gamma_{22}$}
\label{fig_perpetualtsystem}
\end{figure}

\noindent induces a circuit which avoids $t_{cl}$, see Figure \ref{fig_perpetualtsystem}. In both cases the resulting circuit \mbox{avoids $t_{cl}$} and has positive token count. Eventually for a node
$$x_3 \in \gamma_{11} \cap \gamma_{22}$$
the concatenation of the segments
$$(t_{cl},...,x_3) \text{ of } \gamma_{11} \text { and } (x_3,...,t_{cl}) \text{ of } \gamma_{22}$$
induces a circuit which contains $t_{cl}$ and is token-free, because both segments are token-free. In each of the three cases the fundamental property of the \mbox{perpetual $T$-system $(TN,\mu)$} provides a contradiction, which proves the intermediate claim.

\noindent As a consequence the circuit $\gamma$ is elementary with token count
$$\Vert \mu \Vert_{\gamma} \geq \Vert \mu \Vert_\delta = 2,$$
contradicting the fundamental property of perpetual $T$-systems
$$\Vert \mu \Vert_\gamma = \Vert \mu_{cl} \Vert_\gamma =1.$$
\end{enumerate}

\vspace*{-12mm}
\end{proof}
\noindent
The indirect argumentation employed in the proofs of Proposition \ref{prop_tsystemlucent}, \ref{prop_markingequality} and \ref{prop_propagatingenablingequivalence} relies on the same type of contradiction: Construct a reachable marking $\mu^\prime$ and an elementary path $\delta^\prime$ with token \mbox{count $\Vert \mu^\prime \Vert_{\delta^\prime}\geq 2$}. Then apply Proposition \ref{prop_perpetualtsystem}, part 2) respectively Proposition \ref{prop_cpnetsubnetsperpetual}, part 2) to conclude
$\Vert \mu^\prime \Vert_{\delta^\prime}\leq 1$.

\begin{proposition}[Perpetual $T$-systems are lucent]\label{prop_tsystemlucent}
Consider a \mbox{perpetual $T$-system $(TN,\mu_0)$} with regeneration cluster $cl$. For each pair $(\mu_1,\mu_2)$ of reachable markings of $(TN,\mu_0)$ enabling equivalence implies marking equality, i.e.
$$en(TN,\mu_1)=en(TN,\mu_2) \implies \mu_1=\mu_2$$
\end{proposition}
During the indirect proof of Proposition \ref{prop_tsystemlucent} a possible difference between the pair of markings is pinned down to different values at the pre-places of a distinguished transition $t$. The transition is not enabled at neither of the two markings. Due to the liveness of $(TN,\mu_0)$ the missing tokens can be forwarded to the pre-places of $t$ along two token-free paths. One concludes that one of the two paths avoids $t_{cl}$ and can be marked with at least 2 tokens. The result contradicts Proposition \ref{prop_perpetualtsystem}, part 2 ii).
\begin{proof}
The proof is indirect. The assumption $\mu_1\neq \mu_2$ implies the existence of a \mbox{place $p \in TN_P$}, marked \mbox{at $\mu_1$} but unmarked \mbox{at $\mu_2$}. The \mbox{transition $t\in TN_T$} with
$p^\bullet=\{t\}$ satisfies
$$t \notin en(TN,\mu_2).$$
Hence by enabling equivalence
$$t \notin en(TN,\mu_1).$$
As a consequence $t$ has a second pre-place $q \in\ ^\bullet t$ which is unmarked \mbox{at $\mu_1$}, see \mbox{Figure \ref{fig_lucent}}, left and right. W.l.o.g.
$$^\bullet t=\{p,q\}$$
and
$$(\mu_1(p),\mu_1(q))=(1,0) \text{ and } (\mu_2(p),\mu_2(q))=(0,*).$$

\begin{figure}[H]
\centering
\includegraphics[width=10cm]{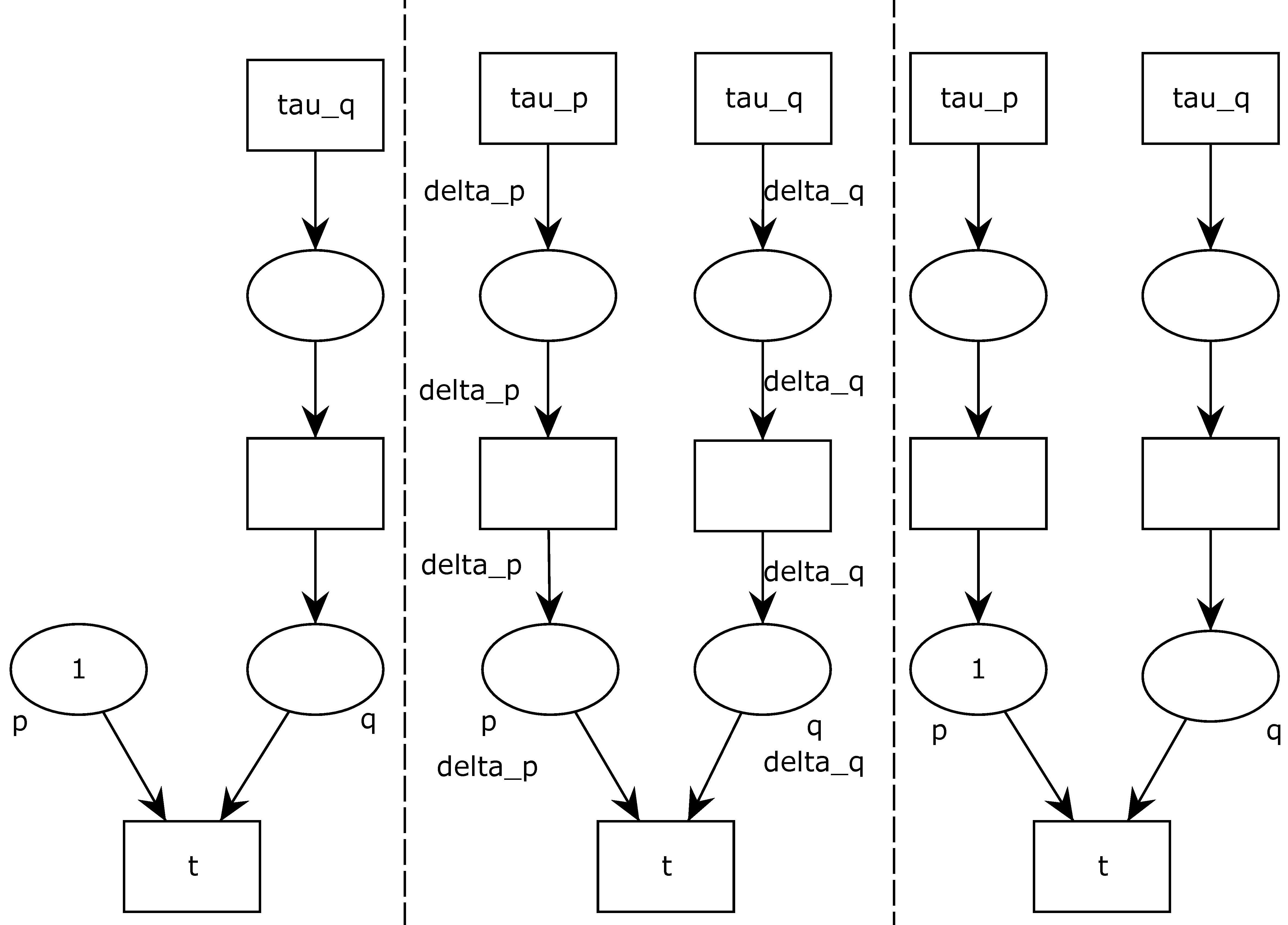}
\caption{ad Prop. \ref{prop_tsystemlucent}: left $\mu_1$,\ middle $\mu_2$,\ right $\mu_1$; ad Prop. \ref{prop_markingequality}: left $\hat{\mu}_1$,\ middle $\hat{\mu}_2$,\ right\ $\hat{\mu}_1$}
\label{fig_lucent}
\end{figure}

\noindent Here the value $\mu_2(q)\in \{0,1\}$ is not yet known. With the notations from Proposition \ref{prop_perpetualtsystem}:
\begin{itemize}
\item \emph{Triple $(\mu_1, \tau_{q}, \delta_q)$}: Proposition \ref{prop_perpetualtsystem}, part 1) provides a pair
$$(\tau_q,\delta_q) \in en(TN,\mu_1)_q \text{ with } \Vert \mu_1 \Vert_{\delta_q}=0$$
see Figure \ref{fig_lucent}, left-hand side.

\item \emph{Pair ($\mu_2,\tau_q$)}: By enabling equivalence
$$\tau_{q} \in en(TN,\mu_1) \implies \tau_{q} \in en(TN,\mu_2)$$

\item  \emph{Triple ($\mu_2,\tau_p,\delta_p$)}: Because $\mu_2(p)=0$ Proposition \ref{prop_perpetualtsystem}, part 1) provides a second pair
$$(\tau_p,\delta_p) \in en(TN,\mu_2)_p \text{ with } \Vert \mu_2 \Vert_{\delta_p}=0$$
see Figure \ref{fig_lucent}, middle.

\item \emph{Triple $(\mu_1,\tau_p,\delta_p)$}: By enabling equivalence
$$\tau_p \in en(TN,\mu_2) \implies \tau_p \in en(TN,\mu_1)$$
In particular $\tau_p \neq t$. Firing $\mu_1 \xrightarrow{\tau_p} \mu^\prime$ implies for the segment $\delta^\prime:=(\tau_p,...,p)$ of $\delta_p$
$$\Vert \mu^\prime \Vert_{\delta^\prime}\geq 2$$
because $\mu_1(p)=1$. Figure \ref{fig_lucent}, right-hand side shows the marking $\mu_1$. But the frozen token due \mbox{to $\mu_1(p)=1$} and the greediness of $cl$ \mbox{ensure $t_{cl} \notin \delta^\prime$}. Hence Proposition \ref{prop_perpetualtsystem}, \mbox{part 2 ii)} implies
$$\Vert \mu^\prime \Vert_{\delta^\prime} \leq 1,$$
a contradiction.\qed
\end{itemize}

\vspace*{-8mm}
\end{proof}

%======================================================================================
\section{Enabling equivalence and marking equality in $CP$-subnets of \\ perpetual free-choice systems}\label{sect_cpsystemsperpetualitymarkings}

In a perpetual free-choice system $(N,\mu_0)$ the adapted $CP$-subnets $\hat{N} \subset N$ and their induced markings have specific properties which are not shared by $CP$-subnets in general live and safe well-formed free-choice systems. These properties derive from the interplay of the regeneration marking and the shutdown sequences of $\hat{N}$.

%-------------------------------------------------------------------------------------------------------------------------------------------------------------------
\begin{proposition}[Token count in adapted $CP$-nets of perpetual free-choice systems]\label{prop_cpnetsubnetsperpetual}
Consider a perpetual free-choice system $(N,\mu_0)$ with a regeneration cluster $cl \subset N$. Let $\hat{N} \subset N$ be a \mbox{a $cl$-adapted} $CP$-subnet.
\begin{enumerate}
\item The $CP$-subnet $\hat{N}$ has no circuits. In particular, each path in $\hat{N}$\ is elementary.
\end{enumerate}
In addition, let $\mu$ be an arbitrary reachable marking \mbox{of $(N,\mu_0)$} and set $\hat{\mu}:=\mu|\hat{N}$.
\begin{enumerate}\addtocounter{enumi}{1}
\item Each path $\delta \subset \hat{N}$ has token count $\Vert \mu \Vert_\delta \leq 1$.

\item Firing a shutdown sequence of $\hat{N}$
$$\mu \xrightarrow{\sigma_{sd}} \mu_{sd}$$
removes all tokens from $\hat{N}$, i.e.
$$\mu_{sd}|\hat{N}=0$$
\item Each transition $t\in \hat{N}$ with a path $(p,...,t) \subset \hat{N},\ p \in \hat{N}_P$ \mbox{and $\hat{\mu}(p)=1$}, can be enabled by firing an occurrence sequence of $(\hat{N},\hat{\mu})$
$$\hat{\mu} \xrightarrow{\sigma} \hat{\mu}_{post} \text{ with } t_{in} \notin \sigma.$$
\end{enumerate}
\end{proposition}

\begin{proof}

\vspace*{-6mm}
\begin{enumerate}
\item For an indirect proof assume the existence of a circuit $\gamma \subset \hat{N}$. Because $\hat{N} \subset N$ is \mbox{a $T$-net} and \mbox{is $cl$-adapted}, for each reachable marking $\mu$ of $(N,\mu_0)$:
$$\Vert \mu \Vert_\gamma=\Vert \mu_{cl} \Vert_\gamma=0$$
Hence each transition from $\gamma$ is dead in $(N,\mu_0)$, a contradiction to the liveness \mbox{of $(N,\mu_0)$}. The second claim follows because a non-elementary path has a node of self-intersection, and the latter produces a circuit.

\item The path $\delta \subset \hat{N}$ extends by concatenation to a path
$$\hat{\delta} =(t_{in},...,t_{out})$$
leading in $\hat{N}$ from the way-in transition $t_{in}$ to a way-out transition $t_{out}$. Due to \mbox{part 1)} the path~$\hat{\delta}$ is elementary. The complement $N \setminus \hat{N}$ is strongly connected. By concatenating $\hat{\delta}$ with an elementary path in the complement leading from $t_{out}$ to $t_{in}$ \mbox{extends $\hat{\delta}$} - and a posteriori also $\delta$ - to an elementary circuit $\delta_N \subset N$. The latter is contained in a $P$-component $C \subset N$, \cite[Analogue of Cor. 5.6]{TV1984}. Due to the fundamental property of perpetual free-choice systems
$$\Vert \mu \Vert_C=1,$$
which implies $\Vert \mu \Vert _\delta \leq 1$.

\item Due to Remark \ref{rem_existencestructuredynamics}, part 3) the free-choice system
$$(\overline{N},\overline{\mu}) \text{ with } \overline{N}:=N\setminus \hat{N} \text{ and } \overline{\mu}:=\mu_{sd}|\overline{N}$$
is live. There exists a reachable marking $\overline{\mu}_{post}$ of $(\overline{N},\overline{\mu})$ which marks all places \mbox{of $\overline{cl}:=\overline{N}\cap cl$}. Because $\overline{N} \subset N$ is place-bordered the extended marking $\mu_{post}$ of $N$ defined as
$$\mu_{post}|\overline{N}:=\overline{\mu}_{post} \text{ and } \mu_{post}|\hat{N} :=\mu_{sd}|\hat{N}$$
is reachable in $(N,\mu_0)$. Because $cl \subset N$ and $\overline{cl} \subset \overline{N}$ have the same places
$$\overline{\mu}_{post} \geq \mu_{\overline{cl}} \implies \mu_{post} \geq \mu_{cl} \implies \mu_{post} = \mu_{cl}$$
which implies
$$\mu_{sd}|\hat{N}= \mu_{post}|\hat{N}=0.$$
\item The claim follows from the previous part because $\hat{N}$ is a $T$-net and $t_{in} \notin \sigma_{sd}$.\qed
\end{enumerate}

\vspace*{-7mm}
\end{proof}

%-------------------------------------------------------------------------------------------------------------------------------------------------------------------
\begin{proposition}[Enabling equivalence and marking equality in adapted $CP$-subnets]\label{prop_markingequality}
Let $(N,\mu_0)$ be a perpetual free-choice system with regeneration cluster $cl \subset N$ and let
$$\hat{N} \subset N$$
be a $cl$-adapted $CP$-subnet. Consider a pair $(\mu_1,\mu_2)$ of reachable markings \mbox{of $(N,\mu_0)$} and assume that the restrictions
$$\hat{\mu}_j:=\mu_j|\hat{N},\ j=1,2,$$
are enabling equivalent, i.e.
$$en(\hat{N},\hat{\mu}_1) = en(\hat{N},\hat{\mu}_2).$$
Then:
\begin{enumerate}
\item \emph{Marking equality on $\hat{N}$}:
$$\hat{\mu}_1=\hat{\mu}_2.$$
\item \emph{Common shutdown sequence}:  Each shutdown sequence of $\hat{N}$ at $\mu_1$
$$\mu_1 \xrightarrow{\sigma} \mu_{1,sd}$$
is also also a shutdown sequence of $\hat{N}$ at $\mu_2$
$$\mu_2 \xrightarrow{\sigma} \mu_{2,sd}$$
\end{enumerate}
\end{proposition}
The idea of the proof of the first statement is the same as for the proof of Proposition \ref{prop_tsystemlucent}. The role of the distinguished transition $t_{cl}$ is now taken by the way-in \mbox{transition $t_{in} \in \hat{N}$}. The argumentation is slightly different: The regeneration cluster $cl$ does not belong \mbox{to $\hat{N}$}, and $t_{in}$ has to be exempted from the transitions under consideration. For the convenience of the reader we therefore give a complete proof. The second result is a simple consequence of the first: A shutdown sequence fires only transitions from $\hat{N}_T\setminus \{t_{in}\}$.
\begin{proof}
Alike to the notation used in the proof of Proposition \ref{prop_perpetualtsystem} and \ref{prop_tsystemlucent} we introduce for a given pair
$$(q,t)\in \hat{N}_P\times{}\hat{N}_T \text{ with }  q ^\bullet=\{t\}$$
the notation
$$en_{sd}(\hat{N},\hat{\mu})_q:=\{(\tau,\delta):\ \tau \neq t_{in}, \tau \in en(\hat{N},\hat{\mu}), \delta=(\tau,...,q,t) \subset \hat{N}\}.$$
It denotes the set of pairs $(\tau,\delta)$ with $\tau\neq t_{in}$ enabled at $\hat{\mu}$ and starting the path $\delta \subset \hat{N}$ to $t$ via $q$.
\begin{enumerate}
\item For an indirect proof of the first part of the Proposition assume
$$\hat{\mu}_1 \neq \hat{\mu}_2$$
There exists a place $p \in \hat {N}_P$, marked \mbox{at $\hat{\mu}_1$} but unmarked \mbox{at $\hat{\mu}_2$}. Consider the well-determined transition $t \in \hat N_T$ with
$p^\bullet=\{t\}$, in particular $t\in \hat{N} \setminus \{t_{in}\}$. The \mbox{transition $t$} is not enabled at $\hat{\mu}_2$. By enabling equivalence
$$t \notin en(\hat{N},\hat{\mu}_2) \implies t \notin en(\hat{N},\hat{\mu}_1).$$
As a consequence $t$ has a second pre-place $q \in\ ^\bullet t$ which is unmarked \mbox{at $\hat{\mu}_1$}. W.l.o.g.
$$^\bullet t=\{p,q\}$$
and
$$(\mu_1(p),\mu_1(q))=(1,0) \text{ and } (\mu_2(p),\mu_2(q))=(0,*).$$
Here the value $\mu_2(q)\in \{0,1\}$ is not yet known. The indirect proof continues along the following steps:
\begin{itemize}
\item \emph{Triple ($\hat{\mu}_1, \tau_{1}, \delta_q$)}: Because
$$\hat{\mu}_1(p) =1 \text{ and } \hat{\mu}_1(q)=0$$
the transition $t$ can be enabled without firing $t_{in}$ by a reachable marking \mbox{of $(\hat{N},\hat{\mu}_1)$} due to \mbox{Proposition \ref{prop_cpnetsubnetsperpetual}}, part 4). Because $\hat{N}$ is a $T$-net Proposition \ref{prop_perpetualtsystem}, part 1) provides a pair
$$(\tau_q,\delta_q) \in en_{sd}(\hat{N},\hat{\mu}_1)_q \text{ with } \Vert \hat{\mu}_1 \Vert_{\delta_q}=0$$
see Figure \ref{fig_lucent}, left-hand side.

\item \emph{Pair ($\hat{\mu}_2,\tau_p)$}: By enabling equivalence
$$\tau_q \in en(\hat{N},\hat{\mu}_1) \implies \tau_q \in en(\hat{N},\hat{\mu}_2)$$

\item \emph{Triple ($\hat{\mu}_2,\tau_p,\delta_p$)}: The transition $\tau_q \in en(\hat{N},\hat{\mu}_2)$ has a pre-place marked \mbox{at $\hat{\mu}_2$}.

\noindent Proposition \ref{prop_cpnetsubnetsperpetual}, part 4), applied to the path $\delta_q$, shows that $t$ is enabled at a reachable marking of
$(\hat{N},\hat{\mu}_2)$. Because $\hat{\mu}_2(p)=0$ Proposition \ref{prop_perpetualtsystem}, part 1) provides a pair
$$(\tau_p,\delta_p) \in en_{sd}(\hat{N},\hat{\mu}_2)_p \text{ with } \Vert \hat{\mu}_2 \Vert_{\delta_p}=0,$$
see Figure \ref{fig_lucent}, middle.

\item \emph{Triple ($\hat{\mu}_1,\tau_p,\delta_p$)}: By enabling equivalence
$$\tau_p \in en(\hat{N},\hat{\mu}_2) \implies \tau_p \in en(\hat{N},\hat{\mu}_1)$$
Figure \ref{fig_lucent}, right-hand side shows $\hat{\mu}_1$. After firing $\hat{\mu}_1 \xrightarrow{\tau_p} \mu^\prime$ the segment
$$\delta^\prime :=(\tau_p,...,p) \text{ of } \delta_p \subset \hat{N}$$
has token count
$$\Vert \mu^\prime \Vert_{\delta^\prime}\geq 2,$$
but due to Proposition \ref{prop_cpnetsubnetsperpetual}, part 2)
$$\Vert \mu^\prime \Vert_{\delta^\prime}\leq 1.$$
The contradiction refutes the assumption of the indirect proof, hence
$$\hat{\mu}_1=\hat{\mu}_2$$
\end{itemize}
\item Due to part 1)
$$\mu_1|\hat{N}=\mu_2|\hat{N}.$$
Remark \ref{rem_existencestructuredynamics} implies: Each marking $\mu_j,\ j=1,2,$ enables a shutdown sequence \mbox{of $\hat{N}$}. Because a shutdown sequence has only transitions from $\hat{N}_T \setminus \{t_{in}\}$, each shutdown sequence enabled at $\mu_1$ is also a shutdown sequence enabled at $\mu_2$, and vice \mbox{ versa.\qed}
\end{enumerate}

\vspace*{-7mm}
\end{proof}

%======================================================================================
\section{Propagating perpetuality and enabling equivalence along \\ $CP$-exhaustions}\label{sect_propagating}

%-------------------------------------------------------------------------------------------------------------------------------------------------------------------
\begin{proposition}[Propagating perpetuality to the complement of an \mbox{adapted $CP$-subnet}]\label{prop_propagatingperpetuality}
Consider a perpetual free-choice \mbox{system $(N,\mu_0)$} with regeneration \mbox{cluster $cl \subset N$}, and a $cl$-adapted $CP$-subnet
$$\hat{N} \subset N \text{ with complement } \overline{N}:=N \setminus \hat{N}.$$
Then for each reachable marking $\mu$ of $(N,\mu_0)$ and for each shutdown \mbox{sequence $\sigma$}
$$\mu \xrightarrow{\sigma} \mu_{sd}$$
for $\hat{N} \subset N$: The free-choice system
$$(\overline{N},\overline{\mu}),\ \overline{\mu}:=\mu_{sd}|\overline{N},$$
is perpetual with regeneration cluster $\overline{cl}:=\overline{N} \cap cl$.
\end{proposition}
%\begin{flushleft}
\noindent Idea of the proof: Both clusters $cl\subset N$ and $\overline{cl} \subset \overline{N}$ have the same places. And each enabled occurrence sequence of
$(\overline{N},\overline{\mu})$ lifts to an enabled occurrence sequence \mbox{of $(N,\mu_{sd})$}.
%\end{flushleft}
\begin{proof}
Remark \ref{rem_existencestructuredynamics}, part 3) implies that $(\overline{N},\overline{\mu})$ is a live and safe free-choice system. We show that
$\mu_{cl}|\overline{N}$ is a regeneration marking of $(\overline{N},\overline{\mu})$:

%\begin{flushleft}
\medskip\noindent For each reachable marking $\overline{\nu}$ of $(\overline{N},\overline{\mu})$ exists a reachable marking $\overline{\nu}_{post}$ of $(\overline{N},\overline{\nu})$ which enables at least one transition and a posteriori - due to the free-choice property of $\overline{N}$ - all transitions of $\overline{cl}$. \mbox{Hence $\overline{\nu}_{post}$} marks all places of $\overline{cl}$. Because $\overline{N} \subset N$ is place-bordered, $\overline{\nu}_{post}$ extends to a reachable marking $\nu_{post}$ of $(N,\mu_0)$ with
%\end{flushleft}
$$\nu_{post}|\overline{N}=\overline{\nu}_{post} \text{ and } \nu_{post}|\hat{N}=\mu_{sd}|\hat{N}$$
Proposition \ref{prop_cpnetsubnetsperpetual} implies
$$\nu_{post}|\hat{N}=\mu_{sd}|\hat{N}=0$$
The clusters
$$cl \subset N \text{ and } \overline{cl} \subset \overline{N}$$
have the same places. Hence
$$\nu_{post} \geq \mu_{cl} \implies \nu_{post}=\mu_{cl}$$
As a consequence $$\overline{\nu}_{post}=\mu_{cl}|\overline{N}=\mu_{\overline{cl}}$$
is a regeneration marking of $(\overline{N},\overline{\mu})$.\qed
\end{proof}

%-------------------------------------------------------------------------------------------------------------------------------------------------------------------
\begin{proposition}[Propagating enabling equivalence to the $T$-net of an \mbox{adapted $CP$-exhaustion}]\label{prop_propagatingenablingequivalence}
Let $(N,\mu_0)$ be a perpetual free-choice system with regeneration \mbox{cluster $cl \subset N$}. Consider a $cl$-adapted $CP$-exhaustion $(\hat{N}_i)_{i \in I}$ of $N$ with the final strongly connected $T$-net
$$\overline{N}:=N \setminus{} \dot\bigcup_{i \in I} \hat{N}_i.$$
Let $(\mu_1,\mu_2)$ be a pair of reachable markings of $(N,\mu_0)$ with
$$en(N,\mu_1)=en(N,\mu_2).$$
\begin{enumerate}
\item For each $i\in I$ the $CP$-subnet $\hat{N}_i \subset N$ has a common shutdown sequence $\sigma_i$ for both markings $\mu_1$ and $\mu_2$. Both markings enable the concatenation
$$\sigma:=\sigma_0*...*\sigma_n,$$
named a \emph{global shutdown sequence}.

\item For $j=1,2$ denote by $\mu_{j,sd}$ the marking of $N$ obtained by firing $\sigma$ at $\mu_j$, i.e.
$$\mu_j \xrightarrow{\sigma} \mu_{j,sd}$$
Then the restrictions to $\overline{N}$
$$\overline{\mu}_{j,sd}:=\mu_{j,sd}|\overline{N}$$
satisfy
$$\overline{\mu}_{1,sd}-\overline{\mu}_{2,sd}=(\mu_1-\mu_2)|\overline{N}.$$
\item The pair of markings
$$(\overline{\mu}_{1,sd},\overline{\mu}_{2,sd})$$
is reachable in the perpetual $T$-system
$$(\overline{N},\mu_{\overline{cl}}) \text{ with } \overline{cl}:=\overline{N} \cap cl$$
and satisfies
$$en(\overline{N},\overline{\mu}_{1,sd})=en(\overline{N},\overline{\mu}_{2,sd}).$$
\end{enumerate}
\end{proposition}

\noindent The idea of the proof is to compare the enabledness of each transition $t \in \overline{N}$ before and after firing $\sigma$. The proof shows: A transition $t$ is enabled \textit{before} firing $\sigma$ at both \mbox{markings $(\mu_1,\mu_2)$} or at none of them if and only \textit{after} firing $\sigma$ the transition is enabled at both markings $(\mu_{1,sd},\mu_{2,sd})$ or at none of them. Besides \mbox{Proposition \ref{prop_perpetualtsystem}} and \ref{prop_markingequality} the main ingredient is the fact that a live and \mbox{bounded $T$-system} is cyclic.

\begin{proof}
We set
$$\overline{\mu}_j:=\mu_j|\overline{N},\ j=1,2.$$
Because $\overline{N} \subset N$ is place-bordered: For each transition $t \in \overline{N}$ and for $j=1,2$ holds
$$t \in en(N,\mu_j) \iff t \in en(\overline{N},\overline{\mu}_j)$$
\begin{enumerate}
\item For each $i \in I$ Proposition \ref{prop_markingequality} provides a common shutdown \mbox{sequence $\sigma_i$} of $\hat{N}_i$ with respect to the markings $\mu_1$ and $\mu_2$. For each pair $i \neq k \in I$ the pre-sets of the non way-in transitions of $\hat{N}_i$ and $\hat{N}_k$ are disjoint, which proves part 1).

\item If $p\in \overline{N} \subset N$ is not a way-in place of $\overline{N}$ then for $j=1,2$
$$\mu_k(p)=\mu_{k,sd}(p).$$
And for a way-in place $p\in \overline{N}_{P,in}$ the change $\mu_{k,sd}(p)-\mu_k(p)$ of both markings depends only on the transitions of $\sigma$. Hence the change is the same whether \mbox{firing $\sigma$} at $\mu_1$ or at $\mu_2$.
\item For $j=1,2$ Proposition \ref{prop_propagatingperpetuality} implies that $(\overline{N},\overline{\mu}_{j,sd})$ is perpetual with regeneration marking
$\mu_{\overline{cl}}$. Live $T$-systems are cyclic, hence $\overline{\mu}_{j,sd}$ is reachable in $(\overline{N},\mu_{\overline{cl}})$. Claim:
$$en(\overline{N},\overline{\mu}_{1,sd})=en(\overline{N},\overline{\mu}_{2,sd})$$
\eject
The proof distinguishes between non-critical transitions and critical transitions.
\begin{flushleft}
i) \emph{Non-critical transitions}: If $t \in \overline{N}_T\setminus \overline{N}_{T,crit}$ then
$$^\bullet t \cap \overline{N}_{P,in}=\emptyset.$$
Hence firing $\sigma$ does not change the marking on $^\bullet t$. As a consequence
\end{flushleft}
$$\left [en(N,\mu_1)=en(N,\mu_2)\right] \implies \left[t \in en(\overline{N},\overline{\mu}_{1,sd}) \iff t \in en(\overline{N},\overline{\mu}_{2,sd})\right]$$
ii) \emph{Critical transitions}: The proof of the claim is indirect. Assume the existence of a critical transition $t \in \overline{N}_{T,crit}$ which violates enabling equivalence, w.l.o.g.
$$t\in en(\overline{N},\overline{\mu}_{1,sd}) \setminus en(\overline{N},\overline{\mu}_{2,sd})$$
Then $t$ is enabled at neither marking $\mu_1$ and $\mu_2$, and firing $\sigma$ at $\mu_1$ enables $t$, but firing $\sigma$ at $\mu_2$ does not. \mbox{Hence $t$} has at least two pre-places, w.l.o.g. $t$ has exactly two pre-places
$$^\bullet t=\{p,q\} \text{ with } p \in \overline{N}_{P,in}$$
satisfying
$$(\mu_1(p),\mu_1(q))=(0,1) \text{ and } (\mu_2(p),\mu_2(q))=(0,0)$$
$$(\overline{\mu}_{1,sd}(p),\overline{\mu}_{1,sd}(q))=(1,1) \text{ and } (\overline{\mu}_{2,sd}(p),\overline{\mu}_{2,sd}(q))=(1,0).$$
\begin{itemize}
\item Proposition \ref{prop_perpetualtsystem}, part 2 i), applied to the perpetual $T$-system $(\overline{N},\mu_{\overline{cl}})$ and its reachable \mbox{marking $\overline{\mu}_{2,sd}$}, provides a pair
$$(\tau,\delta) \in en(\overline{N}, \overline{\mu}_{2,sd})_q,\ \tau \neq t, \text{ with } \Vert \overline{\mu}_{2,sd} \Vert_\delta =0
\text{ and } t_{\overline{cl}} \notin \delta^\prime$$
for the segment
$$\delta^\prime:=(\tau,...,q) \text{ of }\delta.$$
In particular $\tau \in en(\overline{N},\overline{\mu}_{2,sd})$.

\item The transition $\tau$ satisfies
$$\tau \notin en(\overline{N}, \overline{\mu}_{1,sd}):$$
Otherwise, after firing $\overline{\mu}_{1,sd} \xrightarrow{\tau} \mu^\prime$ the path $\delta^\prime$ has token count
$$\Vert \mu^\prime \Vert_{\delta^\prime}\geq 2$$
because $\overline{\mu}_{1,sd}(q)=1$. Then Proposition \ref{prop_perpetualtsystem}, part 2 ii) implies the contradiction
$$\Vert \mu^\prime \Vert_{\delta^\prime}\leq 1$$
\item By the previous step and by the assumed enabling equivalence
$$\tau \notin en(\overline{N}, \overline{\mu}_{1,sd})\implies \tau \notin en(\overline{N},\overline{\mu}_1) \implies \tau \notin en(\overline{N},\overline{\mu}_2)$$
As a consequence
$$\tau \in en(\overline{N},\overline{\mu}_{2,sd}) \setminus en(\overline{N},\overline{\mu}_2),$$
and the transition $\tau$ becomes enabled at a reachable marking of $(\overline{N},\overline{\mu}_2)$ not until firing a way-out transition $t_{out} \in \hat{N}$. Hence $\tau \in \overline{N}_{T,crit}$ is a further critical transition with
$$\tau \in en(\overline{N},\overline{\mu}_{2,sd}) \setminus en(\overline{N},\overline{\mu}_{1,sd}).$$
\item The previous result implies that at least one pre-place $\tau$ is marked at $\overline{\mu}_{1,sd}$ due to the firing of $t_{out}$. As a consequence
$$t_{\overline{cl}} \neq t:$$
Otherwise the enabling $t \in en(\overline{N},\overline{\mu}_{1,sd})$, the greediness of $\overline{cl}$, and the fact, that places of \mbox{a $T$-net} do not branch, imply $t=\tau$, which has been excluded above.
\end{itemize}

\noindent We now iterate the whole argument above: It derives from the critical transition
$$t_1:=t \in \overline{N}_{T,crit} \cap \left(en(\overline{N},\overline{\mu}_{1,sd}) \setminus en(\overline{N},\overline{\mu}_{2,sd})\right)$$
a second critical transition
$$t_2:=\tau \in \overline{N}_{T,crit} \cap \left(en(\overline{N},\overline{\mu}_{2,sd}) \setminus en(\overline{N},\overline{\mu}_{1,sd})\right)$$
and a path
$$\delta_1:=\delta=(t_2,...,q_1,t_1) \subset \overline{N} \text{ with } \Vert \overline{\mu}_{1,sd} \Vert_{\delta_1}=1,\ \Vert \overline{\mu}_{2,sd} \Vert_{\delta_1}=0
\text{ and } t_{\overline{cl}} \notin \delta_1$$
see Figure \ref{fig_enablingcriticaltransitions}.

\medskip After finitely many steps we obtain a family of critical transitions
$$t_k \in \overline{N}_{T,crit}, \ k=1,...,m,$$
and elementary paths $$\delta_k=(t_{k+1},...,q_k,t_k) \subset \overline{N},\ k=1,...,m-1,\ t_{cl} \notin \delta_k,$$
satisfying
$$\Vert \overline{\mu}_{1,sd} \Vert_{\delta_k}=\begin{cases}
1 & \text{if } k \text{ odd} \\
0 & \text{if } k \text{ even} \end{cases},\ \Vert \overline{\mu}_{2,sd} \Vert_{\delta_k}=\begin{cases}
1 & \text{if } k \text{ even} \\
0 & \text{if } k \text{ odd} \end{cases}$$
Because $\overline{N}$ has only finitely many critical transitions, a subset of these paths concatenates and induces a circuit
$$\gamma \subset \overline{N} \text{ with } t_{\overline{cl}} \notin \gamma.$$
The circuit $\gamma$ satisfies
$$\Vert \mu_{\overline{cl}} \Vert_\gamma= \Vert \overline{\mu}_{1,sd} \Vert_\gamma= \Vert \overline{\mu}_{2,sd} \Vert_\gamma \geq 1$$
Hence $\gamma$ has at least one elementary subcircuit which is marked at $\mu_{\overline{cl}}$. The result contradicts the fundamental property of perpetual $T$-systems because the subcircuit does not contain $t_{\overline{cl}}$.

\begin{flushleft}
The contradiction refutes the assumption that the critical transition $t=t_1$ violates enabling equivalence. Hence all critical transitions $t \in \overline{N}_{T,crit}$ satisfy
\end{flushleft}
$$t \in en(\overline{N},\overline{\mu}_{1,sd}) \iff t \in en(\overline{N},\overline{\mu}_{2,sd})$$

\noindent iii) \emph{Enabling equivalence}: The two previous parts show
$$en(\overline{N},\overline{\mu}_{1,sd}) = en(\overline{N},\overline{\mu}_{2,sd}),$$
which finishes the proof of the Proposition.\qed
\end{enumerate}

\vspace*{-7mm}
\end{proof}

\begin{figure}[!h]
\vspace*{-2mm}
\centering
\includegraphics[width=10cm]{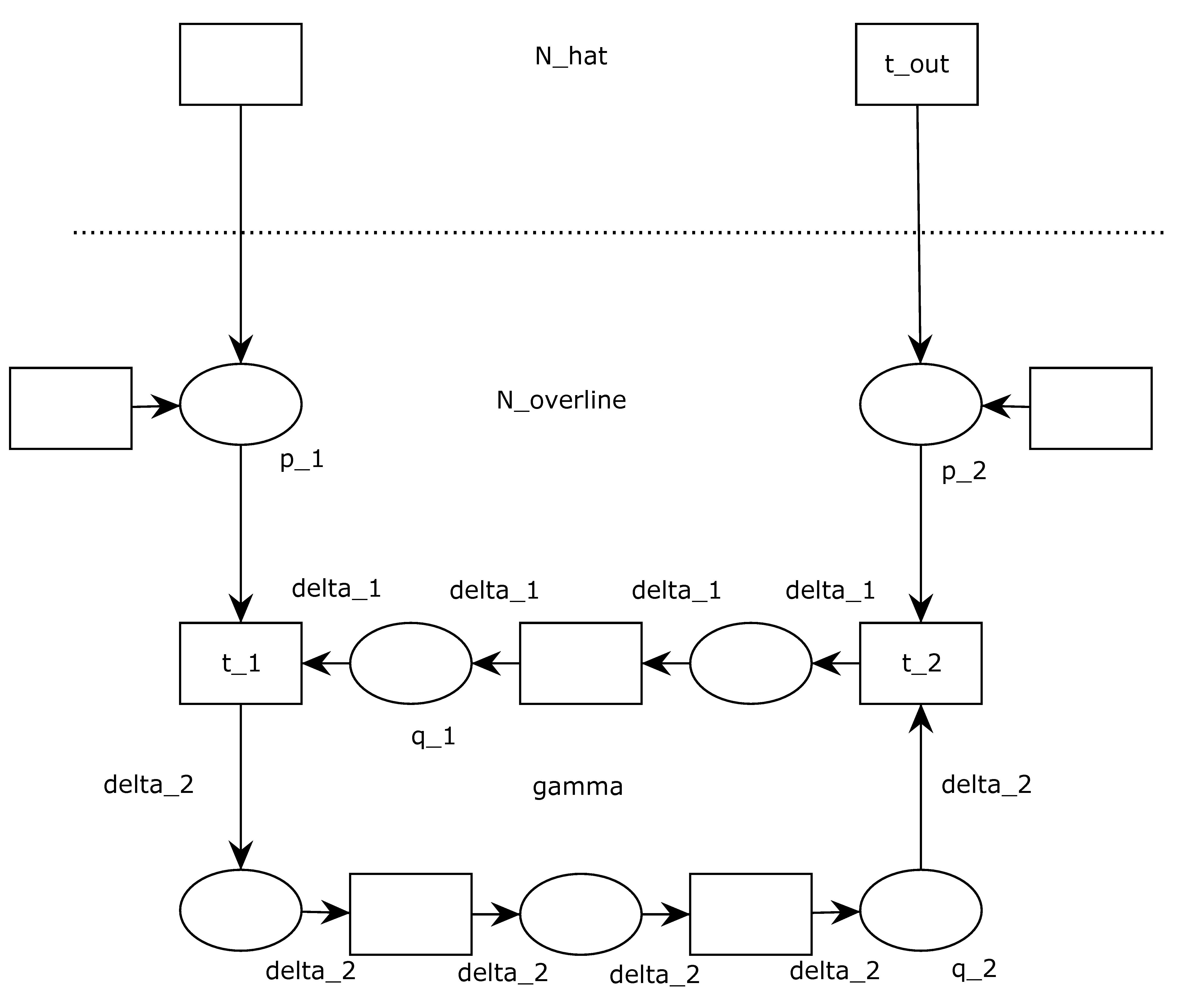}\vspace*{-2mm}
\caption{Way-in places $p_j$, critical transitions $t_j$, paths $\delta_j$, particular case $t_3=t_1$}
\label{fig_enablingcriticaltransitions}\vspace*{-3mm}
\end{figure}

%======================================================================================
\section{Statement and proof of van der Aalst's theorem}\label{sect_proofvanderaalsttheorem}
\begin{theorem}[Van der Aalst's theorem on lucency from \cite{Aal2018}]\label{theor_vanaalsttheorem}
Each perpetual free-choice system is lucent.
\end{theorem}

\begin{proof} Consider a perpetual free-choice system $(N,\mu_0)$ with a regeneration \mbox{cluster $cl \subset N$} and its regeneration marking $\mu_{cl}$.
\begin{enumerate}
\item \emph{$CP$-exhaustion}: Theorem \ref{theor_existencecpexhaustion} provides a $cl$-adapted
$CP$-exhaustion of $N$
$${(\hat{N}_i)}_{i \in I},\ I=\{0,...,n\} \subset \mathbb{N},$$
with $CP$-subnets $\hat{N}_i \subset N$, complements
$$\overline{N}_i:=\overline{N}_{i-1}\setminus \hat{N}_i,\ \overline{N}_{-1}:=N,$$
and the final strongly connected $T$-net
$$\overline{N}:=N\setminus \dot\bigcup_{i\in N} \hat{N_i}$$
The disjoint union of full subnets of $N$
$$N_{exh}:= \overline{N}\ \dot\cup\ \hat{N}_0\ \dot\cup \ ...\ \dot\cup \ \hat{N}_n$$
is a subnet of $N$ with the same nodes as $N$. The regeneration \mbox{cluster $cl$} intersects each complement
$\overline{N}_i,\ i \in I,$ and $\overline{N}$ in a non-empty cluster. Set
$$\overline{cl}:= \overline{N} \cap cl \subset \overline{N}.$$
\end{enumerate}\bigskip
\begin{flushleft}
To continue the proof assume a pair of reachable markings $(\mu_1,\mu_2)$ of $(N,\mu_0)$ with
\end{flushleft}
$$en(N,\mu_1)=en(N,\mu_2).$$
\begin{enumerate}[start=2]
\item \emph{Marking equality in the $CP$-subnets}:
Proposition \ref{prop_markingequality} implies for \mbox{each $i \in I$}
$$(\mu_1-\mu_2)|\hat{N}_i=0$$
\item \emph{Marking equality in the final $T$-net}: Proposition \ref{prop_propagatingenablingequivalence} considers simultaneously the collection of all $CP$-subnets  $\hat{N}_i,\ i \in I$. The proposition provides a global shutdown sequence $\sigma$ enabled at both markings $\mu_j,\ j=1,2,$
$$\mu_j \xrightarrow{\sigma} \mu_{j,sd}$$
such that the resulting markings of $\overline{N}$
$$\overline{\mu}_{j,sd}:=\mu_{j,sd}|\overline{N}$$
\begin{itemize}
\item are reachable in the $T$-system $(\overline{N},\mu_{\overline{cl}})$,
\item satisfy
$$\overline{\mu}_{1,sd}-\overline{\mu}_{2,sd}=(\mu_1-\mu_2)|\overline{N}$$
\item and are enabling equivalent
$$en(\overline{N},\overline{\mu}_{1,sd})=en(\overline{N},\overline{\mu}_{2,sd}).$$
\end{itemize}

\noindent Proposition \ref{prop_propagatingperpetuality} implies that $(\overline{N},\mu_{\overline{cl})})$ is perpetual. Proposition \ref{prop_tsystemlucent} concludes
$$\overline{\mu}_{1,sd}-\overline{\mu}_{2,sd}=0.$$
As a consequence
$$(\mu_{1} - \mu_{2})|\overline{N}=0$$
Combining part 2) and 3) of the proof shows
$$\mu_1=\mu_2$$
and finishes the proof.\qed
\end{enumerate}

\vspace*{-6mm}
\end{proof}

%======================================================================================
\section{Concluding remarks}\label{sect_finalremarks}

$CP$-subnets of a free-choice net have been introduced by Desel and Esparza. In \cite{Wehler2010} we used $CP$-nets to show the theorem of Gaujal, Haar and Mairesse about the existence of unique blocking markings in live and bounded free-choice systems.

The present proof of Theorem \ref{theor_vanaalsttheorem} does not presuppose the blocking theorem but it makes a similar use of $CP$-nets. As van der Aalst remarks, for a perpetual free-choice system both his theorem and the uniqueness part of the blocking theorem give the same result when applied to those reachable markings, which enable the transitions of one single common cluster but no other transitions.\bigskip

Our proof uses in an essential way the existence of adapted $CP$-exhaustions for well-formed free-choice nets. Figure \ref{fig_planofproof} visualizes the logical structure of the proof. The figure indicates the results referring to the building blocks of \mbox{the $CP$-exhaustion} and their relations. If the underlying net lacks well-formedness the proof does not apply. Therefore it would be interesting to isolate those consequences of regeneration clusters \mbox{in $T$-systems} which underly the constructions from Section \ref{sect_tsystemsperpetualitymarkings} and notably Figure \ref{fig_lucent}:
\begin{itemize}
\item Marking an unmarked pre-place of a transition by forwarding tokens along an elementary, initially token-free path, which starts at an enabled transition.
\item Markings which are enabling equivalent but distinct create a distinguished elementary path: The path has token count at least two and avoids the regeneration cluster.
\item The greediness of the regeneration cluster ensures the safeness of each elementary path which avoids the regeneration cluster.
\end{itemize}
How do these properties generalize in a direct manner from perpetual $T$-systems to perpetual free-choice systems - without using the $CP$-exhaustion? How far can one relax the assumptions of van der Aalst's theorem and still prove lucency?

\mbox{In \cite[Theor. 3, FN 2]{Aal2020}} van der Aalst mentions that he currently investigates his theorem in this direction.

%-------------------------------------------------------------------------------------------------------------------------------------------------------------------
\bibliographystyle{fundam}
%\bibliography{citations}

%======================================================================================

\end{document}